\documentclass[aps,pra,onecolumn,nobibnotes,superscriptaddress,nofootinbib]{revtex4-2}
\usepackage{graphicx} 
\usepackage{caption,graphics,physics}
\usepackage{subcaption}
\usepackage{amsmath,xcolor,bm}
\usepackage{qcircuit}
\usepackage{mathtools,amsmath,amsthm,amsfonts,amssymb,amscd}
\usepackage{hyperref}
\usepackage[ruled,vlined]{algorithm2e}
\usepackage{algpseudocode}

\usepackage[capitalize]{cleveref}

\hypersetup{
    colorlinks=true,
    linkcolor=red,
    filecolor=magenta,      
    urlcolor=cyan,
}
\usepackage{geometry}
\usepackage{makecell}

\newtheorem{theorem}{Theorem}

\newtheorem{lemma}[theorem]{Lemma}

\newtheorem{proposition}[theorem]{Proposition}

\newtheorem{definition}[theorem]{Definition}

\newtheorem{problem}{Problem}
\newcommand{\bs}{\boldsymbol}

\Crefname{appsec}{Appendix}{Appendices}

\begin{document}

\title{Quantum random power method for ground state computation}

\author{Taehee Ko}
\email{kthmomo@kias.re.kr}
\affiliation{School of Computational Sciences, Korea Institute for Advanced Study}

\author{Hyowon Park}
\email{hyowon@uic.edu}
\affiliation{Department of Physics, University of Illinois at Chicago}

\author{Sangkook Choi}%
\email{sangkookchoi@kias.re.kr}
\affiliation{School of Computational Sciences, Korea Institute for Advanced Study}

\begin{abstract}
 
  We present a quantum-classical hybrid random power method that approximates a ground state of a Hamiltonian. The quantum part of our method computes a fixed number of elements of a Hamiltonian-matrix polynomial via quantum polynomial filtering techniques with either  Hamiltonian simulation or block encoding. The use of the techniques provides a computational advantage that may not be achieved classically in terms of the degree of the polynomial.  The classical part of our method is a randomized iterative algorithm that takes as input the matrix elements computed from the quantum part and outputs an approximation of ground state of the Hamiltonian. We prove that with probability one, our method converges to an approximation of a ground state of the Hamiltonian, requiring a constant scaling of the per-iteration classical complexity. The required quantum circuit depth is independent of the initial overlap and has no or a square-root dependence on the spectral gap. The iteration complexity scales linearly as the dimension of the Hilbert space when the quantum polynomial filtering corresponds to a sparse matrix. We numerically validate this sparsity condition for well-known model Hamiltonians. We also present a lower bound of the fidelity, which depends on the magnitude of noise  occurring from quantum computation regardless of its charateristics, if it is smaller than a critical value. Several numerical experiments demonstrate that our method provides a good approximation of ground state in the presence of systematic and/or sampling noise. 
\end{abstract}
\maketitle

\section{Introduction}

Quantum computing algorithms are shown as a promising means for various problems in theoretical sciences \cite{o2016scalable,veis2010quantum,mcardle2020quantum,o2019quantum,lewis2024improved}. Among those, the computation of ground state and the corresponding energy is of significant importance in physics, chemistry, and material science. Despite the importance of the problem, it is not clear whether there exists an efficient algorithm even when using the power of quantum computations (i.e. the QMA-hard result  \cite{kitaev2002classical,aharonov2009power,kempe2006complexity,oliveira2005complexity}). Nevertheless, under additional assumptions, algorithms that exploit quantum resource or are inspired conceptually have been demonstrated for ground-state energy estimation. 
For instance, quantum algorithms for ground-state energy estimation were proposed based on a Krylov subspace method, given access to Hamiltonian simulation or block encoding \cite{kirby2023exact,seki2021quantum,lee2023sampling,motta2020determining}. For the same task, the recent work \cite{gharibian2022dequantizing} proposed dequantizing algorithms, assuming that the initial overlap with ground state is sufficiently large.

Beyond the estimation of ground state energy, quantum algorithms for preparing ground state have been developed based on two popular frameworks: the quantum signal processing (QSP) or quantum singular value transformation (QSVT) \cite{low2017optimal,GSLW19} and the variational quantum eigensolver (VQE) \cite{peruzzo2014variational,grimsley2019adaptive,mcclean2016theory,larocca2023theory}. In the former case, state of the art includes polynomial-filtering-based algorithms by Lin and his colleagues \cite{lin2020near,dong2022ground}. In the regime of early fault-tolerant quantum computing, those algorithms prepare a ground state with rigorous guarantee and a good initial state asssumed. On the other hand, the VQE is more suitable on near-term quantum devices due to its variational structure with a choice of ansatz. However, no convergence guarantee is yet rigorously found, except results on the trainability of VQE \cite{larocca2023theory,ragone2024lie}.

The ground state problem is viewed as a constrained optimization problem. Due to the exponential growth of the dimension of the Hilbert space, the problem of practical interest is considered as a large-scale optimization task. In classical algorithms, the notion of \emph{randomization} has been employed  for large-scale optimizations as an effective algorithmic technique. Examples are the Karzmacz method for solving a linear system \cite{strohmer2009randomized}, quantum Monte Carlo (QMC) for ground-state problem \cite{becca2017quantum,ceperley1986quantum,lu2020full}, the stochastic gradient descent for training machine learning models \cite{bottou2018optimization} and the principal component analysis (PCA) \cite{shamir2016convergence}, and the random coordinate descent for convex optimization \cite{nesterov2012efficiency}. Recently, the idea of using randomization has also found useful for quantum computing applications, such as  optimization of parameterized quantum circuits \cite{ding2023random,sweke2020stochastic}, Hamiltonian simulation \cite{campbell2019random,childs2019faster} and extracting the properteis of the ground states and the Gibbs states \cite{wang2024qubit}. The common feature of randomized methods is a trade-off between  computational resources and the sampling cost. Interestingly, it is often observed that  the randomized approach is more efficient than the deterministic counterpart theoretically and empirically.

In this paper, we present the first quantum version of randomized power method that yields an approximation of ground state as a classical vector. Our algorithm involves two main subroutines: the estimation of elements of a matrix polynomial of the Hamiltonian using  quantum computing techniques and a  randomized iteration based on the power method on a classical computer. In the former subroutine, a fixed number of elements of a matrix polynomial of the Hamiltonian are sampled using a quantum polynomial filtering technique with either a linear combination of unitaries (LCU) representation based on the Fourier-series approximation or the QSVT techniques \cite{GSLW19,dong2021efficient}.  Then, the sampled matrix elements are used to construct a sparse stochastic gradient and update the iteration, which we call random power method. The workflow of our algorithm is illustrated in \cref{fig:QRPM}. 

\begin{figure}[htbp]
    \centering
    \includegraphics[scale=0.6]{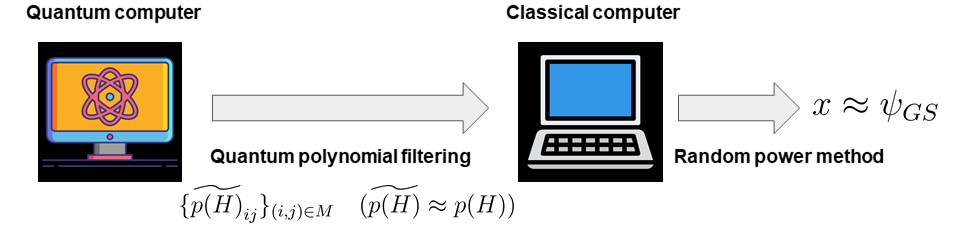}
    \caption{An illustration of quantum random power method (QRPM). At every iteration, a set of indices, $M$, is sampled randomly, and then the corresponding elements of a matrix polynomial filter, $\{\widetilde{p(H)}_{ij}\}_{(i,j)\in M}$, are computed from a quantum computer. We call this quantum polynomial filtering. In practice, matrix elements obtained from the quantum polynomial filtering may be biased due to various types of noise (e.g. systematic error from the Trotter formula, sampling noise from a Hadamard test, and so on). This implies that we essentially estimate elements of a perturbation of the true matrix polynomial, $p(H)$, from a quantum computer. Next, we use the matrix elements as input for a classical randomized algorithm, which we call random power method. After a sufficient number of iterations, the method yields a classical vector, $\bm x$, close to a ground state of $H$, $\psi_{GS}$.}
    \label{fig:QRPM}
\end{figure}

\section*{Contribution}
This paper focuses on developing a hybrid randomized version of the power method for ground-state computation. On the algorithmic and theoretical sides, the main contributions of this work are summarized as follows,
\begin{itemize}

    \item We show that our algorithm converges to a ground state of a given Hamiltonian $H\in\mathbb{C}^{N\times N}$ to a target precision $\epsilon$ with probability one  (\cref{thm: convergence2}).  Specifically, for a desired precision $\epsilon$, if the initial guess $\bm x_1$ has a non-zero overlap $\gamma$ with a ground state of $H$, the algorithm outputs an approximation of a ground state  whose fidelity with the ground state is greater than $1-\epsilon$. To the best of our knowledge, this is the first theoretical result that guarantees an approximation of ground state with probability \emph{one}.  For our analysis, we employ a submartingale approach to quantify the convergence rate of  fidelity. The convergence analysis implies that the number of iterations scales as $\mathcal{O}(\frac{Ns^2}{\epsilon^2}\log\frac{1}{\gamma})$, when  the matrix polynomial from quantum polynomial filtering is $s$-sparse. 
    

    \item We propose a simple and efficient gradient estimator for updating the iteration in our method, which requires a $\mathcal{O}(1)$ classical time and a polylogarithmic quantum circuit complexity with system size.

    \item We present a lower bound of the fidelity of an approximation of ground state with the true one in \cref{lem: overlap1}. The lower bound depends linearly on the magnitude of noise between their corresponding matrices, if it is smaller than some critical value. The value is formulated by the spectral gap of the true matrix.  

    \item We provide several numerical results demonstrating that our method outputs a fairly good approximation of a ground state and supporting our convergence analysis. A part of the results also shows that the noise condition in \cref{lem: overlap1} holds in many cases.  

    \item  The proposed method can be straightforwardly generalized to the task of approximating an excited state and the corresponding eigenvalue by changing the quantum polynomial filtering accordingly.

\end{itemize}

\section*{Comparison to existing algorithms (changed from related work)}

\cref{table: precomparison} summarizes the overall comparison of our quantum random power method ( abbreviated as QRPM) with classical and quantum algorithms. In QRPM, we assume that $p(H)$ is a sparse matrix and we sample non-zero matrix elements using the sparsity. As will be discussed in \cref{subsec: sparsity}, this assumption is reasonable for well-known model Hamiltonians.    

Compared to their deterministic counterparts, classical randomized algorithms have found efficient in finding an eigenvector corresponding to the smallest or the largest eigenvalue. For example, QMC algorithms \cite{becca2017quantum,ceperley1986quantum} involve adaptive rules of iteration update with randomized techniques. These algorithms as analyzed in \cite{lu2020full} have an inverse-polynomial scaling dependence on the initial overlap and at least a square dependence on the spectral gap in terms of the iteration complexity. However, our algorithm with LCU has a logarithmic scaling of the initial overlap dependence with a square dependence on the spectral gap as shown in \cref{table: precomparison}.  In machine learning applications,  a randomized version of the power method for the PCA is proven to converge \cite{shamir2016convergence,xu2018accelerated}. Compared to the stochastic power methods \cite{xu2018accelerated}, our method with LCU achieves a $\mathcal{O}(\Delta^{-2})$ scaling of the iteration complexity in terms of spectral gap, which is comparable to their methods, but with a constant scaling of the per-iteration classical time while that of their methods depend linearly on the dimension. A clearer comparison may be made with the classical power method when it is applied to quantum sparse Hamiltonians as represented by a sum of $\mathcal{O}(\text{poly}(n))$ Pauli strings.  As in \cref{table: precomparison}, the classical power method requires per-iteration complexity scaling linearly as the size of Hamiltonian at each iteration, while our algorithm requires only a constant scaling of per-iteration complexity. Nevertheless, classical power method and our algorithm have a comparable scaling of the iteration complexity in terms of the dimension $N$. In addition, our algorithm with  QSVT \cite{GSLW19} or QETU \cite{dong2021efficient} has a quadratically improved dependence of the spectral gap in terms of quantum complexity, while with LCU, it has a quadratically-worse dependence of the spectral gap than the classical power method.

In quantum algorithms, the task of ground state preparation is to build a quantum circuit that drives an initial state to a ground state approximately, such as VQE \cite{peruzzo2014variational,grimsley2019adaptive} and QSP-based methods  \cite{GSLW19,lin2020near,dong2022ground}. The convergence of VQE is yet clearly understood. In the latter case, those algorithms \cite{lin2020near,dong2022ground} are theoretically analyzed in terms of query complexity and query depth (e.g. the query complexity for Hamiltonian simulation or a block encoding of $H$). In particular, the algorithms in \cite{dong2022ground} involve an inverse polynomial scaling of the initial overlap, $\mathcal{O}(\frac{1}{\gamma})$ or $\mathcal{O}(\frac{1}{\gamma^2})$, in the query complexity, which is similar to QPE. In general, preparing a good initial state (e.g. $\gamma=\mathcal{O}(\frac{1}{\text{poly}(n)})$) may be hard in quantum chemistry theoretically \cite{schuch2009computational} and numerically  \cite{fomichev2024initial,mcclean2014exploiting} as pointed out in \cite{gharibian2022dequantizing}.  In this viewpoint, in the general case where $\gamma=\mathcal{O}(b^{-n})$ for some $b>1$, the query depth of the algorithms \cite{dong2022ground} scales exponentially, while that of our algorithm requires a constant scaling of query depth in terms of $\gamma$, although both algorithms require exponential scalings of the query complexity. Additionally, the algorithms \cite{dong2022ground} have a $\mathcal{O}(\Delta^{-1})$ scaling of query depth, while our algorithm with QETU requires a $\mathcal{O}(\Delta^{-\frac{1}{2}})$ scaling depth. This is because
the filtering polynomial in \cite{dong2022ground} needs to seperate the ground state energy from the other excited-state energies, whereas in our algorithm, a filtering polynomial allows for some contributions of the excited-state energies, which is a more relaxed filtering condition.



\begin{table}[htbp]

        \begin{tabular}{| c | c| c| c| c| c|}
        \hline 
        Method  & CPM   & \cite[Theorem 6]{dong2022ground}  &  \cite[Theorem 11]{dong2022ground}  & QRPM(LCU) & QRPM(QETU)  \\
  \hline Per-iteration classical time & $\mathcal{\tilde{O}}(N)$ & NA & NA & $\mathcal{O}(\frac{1}{\Delta^2})$ & $\mathcal{O}(1)$  \\
  \hline Per-iteration memory 
  & $\mathcal{O}(N)$ & NA & NA & $\mathcal{O}(1)$ & $\mathcal{O}(1)$  \\
  \hline Iteration number & $\mathcal{O}(\frac{1}{\Delta}\log\frac{1}{\gamma\epsilon})$ &  NA & NA & $\mathcal{\tilde{O}}(\frac{N}{\epsilon^{2}}\log \frac{1}{\gamma})$ & $\mathcal{\tilde{O}}(\frac{N}{\epsilon^{2}}\log \frac{1}{\gamma})$  \\
  \hline Iteration complexity & $\mathcal{\tilde{O}}(\frac{N}{\Delta}\log\frac{1}{\gamma\epsilon})$ & NA & NA & $\mathcal{\tilde{O}}(\frac{N}{\Delta^2\epsilon^{2}}\log \frac{1}{\gamma})$ & $\mathcal{\tilde{O}}(\frac{N}{\epsilon^{2}}\log \frac{1}{\gamma})$ \\
  \hline Query complexity & NA & $\widetilde{\mathcal{O}}(\frac{1}{\gamma^2\Delta}\log\frac{1}{\gamma\epsilon})$ & $\widetilde{\mathcal{O}}(\frac{1}{\gamma\Delta}\log\frac{1}{\gamma\epsilon})$ & $\mathcal{\tilde{O}}(\frac{N}{\epsilon^{2}\Delta^2}\log \frac{1}{\gamma})$ & $\mathcal{\tilde{O}}(\frac{N}{\epsilon^{2}\sqrt{\Delta}}\log \frac{1}{\gamma})$ \\
  \hline Query depth & NA  & $\mathcal{O}(\frac{1}{\Delta}\log\frac{1}{\gamma\epsilon})$ & $\mathcal{O}(\frac{1}{\gamma\Delta}\log\frac{1}{\epsilon})$ & $\mathcal{O}(1)$ & $\mathcal{O}(\frac{1}{\sqrt{\Delta}})$ \\
  \hline Success probability & 1  & $\geq\frac{2}{3}$ &  $\geq\frac{2}{3}$ & 1 & 1\\
  \hline Ancilla qubit number & NA  & $1$ & $2$ & $1$  
   &  $2$  \\
  \hline
   \end{tabular}
    
     \caption{Comparison of our algorithms (QRPM) to the classical power method, abbreviated as CPM, and two algorithms in \cite{dong2022ground} for ground state computation or preparation in terms of classical complexity such as per-iteration classical time, per-iteration memory, iteration number, iteration complexity (per-iteration classical time $\times$ iteration number) and quantum complexity such as query complexity and query depth if possible (if not, indicated as NA). Additionally, the algorithms are evaluated in terms of success probability and the number of qubits required. $\gamma$ is the overlap between the initial guess and the ground state, $\Delta$ is a lower bound of the spectral gap, and $1-\epsilon$ is the target fidelity.  For  comparison of our algorithms with the ones in \cite{dong2022ground}, we assume as input model the Hamiltonian simulation in this table.   }
     \label{table: precomparison}
\end{table}

\section*{Organization}
The rest of the paper is organized as follows. We  present our main result and algorithmic components for solving the ground state problem in \cref{sec: prem}. \cref{sec: algorithm} shows convergence guarantee of our algorithm, with emphasis on convergence in probability and the implication to the iteration complexity. Additionally, we provide a perturbation theorem that elucidates the effect of biases occuring from the quantum filtering technique on the accuracy of approximate ground state.  
In \cref{sec: numercial result}, several numerical experiments validate the performance of our algorithm for approximating ground state for different model Hamiltonians. In the appendixes,
we supplement analytical details and all proofs
thereof.

\section{Preliminaries}\label{sec: prem}
In this paper, $H$ denotes the Hamiltonian matrix. We use $\norm{\cdot}$ for the 2-norm of a vector of a matrix.  The state $\ket{\bm x}$ in the braket notation is always a unit vector, $\norm{\bm x}=1$, and a vector without the notation, $\bm x$, is not necessarily a unit vector. This notation will be used mainly in analytic results in the appendixes. For the inner product of two vectors, if they are unit vectors, we denote it as $\bra{\bm v}\ket{\bm w}$. Otherwise, we use $(\bm v,\bm w)$, the conventional notation in linear algebra.

This paper focuses on the approximation of a ground state of a Hamiltonian to the desired precision $\epsilon$ using quantum and classical resources.  Specifically, the problem is set up as follows
\begin{problem}[Finding an approximation of a ground state with probability one]\label{problem}For a given precision $\epsilon>0$ and a given Hamiltonian $H\in\mathbb{C}^{N\times N}$, design a randomized hybrid algorithm that outputs a classical description of state vector $\ket{\bm x}$ satisfying 
\begin{equation}\label{eq: fidelity}
   \sqrt{\sum_{j=1}^g\abs{\bra{ \psi_{GS,j}}\ket{\bs x}}^2}\geq 1-\epsilon\quad \text{with probability }1.
\end{equation}Here the set, $\{\ket{\psi_{GS,j}}\}_{j=1}^g$, is the collection of degenerate ground states of $H$, all of which are solutions of the constrained optimization problem, 
\begin{equation}\label{GSproblem}
    \min_{\norm{\bs x}=1}\bra{\bs x}H\ket{\bs x}.
\end{equation}  
\end{problem}The quantity \eqref{eq: fidelity} is the fidelity of the iterate $\bm x$ with the subspace of ground states of $H$, which measures how much close the iterate is to some ground state of $H$. In the case of non-degnerate ground state, the quantity \eqref{eq: fidelity} becomes the usual fidelity $\abs{\bra{\psi_{GS}}\ket{\bm x}}$. More generally, \eqref{eq: fidelity} covers the case of degenerate ground states when the eigenvalues of $H$ are ordered as follows
\begin{equation}\label{eigvals}
    \lambda_1=\cdots=\lambda_g<\lambda_{g+1}\leq\cdots\lambda_N
\end{equation}for some $g\geq 2$. For the analysis, we first work with the case of non-degenerate ground state ($g =1$), and extend it to the case of degenerate ones ($g\geq 2$).

\subsection{Matrix-element estimation from  quantum polynomial filtering}\label{sec: matrix-element estimation}

A main component in our algorithm is the estimation of the matrix element of the Hamiltonian matrix polynomial $p(H)$, for example, for any indices $i,j\in[N]$,
\begin{equation}
    \bra{ i}p(H)\ket{ j},
\end{equation}where $\ket{ i}$ denotes the computational basis element corresponding to the integer $i$, and the same for $\ket{ j}$. To estimate this quantity, one can consider two approaches: one based on the Hamiltonian evolution \cite{lin2022heisenberg} and the other based on the quantum singular value transformation (QSVT) \cite{GSLW19}. 

Hamiltonian evolution is built on the ability to implement $e^{iHt}$. We estimate a matrix polynomial of $H$ using a Fourier-series approximation. By the property of the Fourier expansion, we obtain the following approximation for any smooth function $p(x)$ on $[-\pi,\pi]$ and some sufficiently large $M$,
\begin{equation}
    p(x) \approx \sum_{m=-M}^Mc_me^{imx}
\end{equation}The coefficient $c_m$ depends on the choice of a kernel \cite{weisse2006kernel}. In our applciations, we used the Fe\'{j}er kernel as discussed in \cref{sec: fejer}. Using this approximation and Hamiltonian simulation, we obtain access to a matrix polynomial of $H$ approximately,
\begin{equation}
    p(H) \approx \sum_{m=-M}^Mc_me^{iHm}.
\end{equation}Especially, estimating matrix elements of $p(H)$ and the corresponding elements of the exponentials of $H$ are equivalent as follows,
\begin{equation}\label{eq: esimatemat}
    \bra{\bm i}p(H)\ket{\bm j} \approx 
    \sum_{m=-M}^Mc_m\bra{\bm i}e^{iHm}\ket{\bm j}.
\end{equation}
To compute this matrix element, we use the Hadamard test in \cref{fig:circuit}. 

Alternatively, one can use  the QSVT with a block encoding of $H$ \cite{GSLW19} or the quantum eigenvalue transformation of unitaries (QETU) \cite{dong2021efficient} with Hamiltonian simulation to compute matrix elements. Similar to the above LCU approach \eqref{eq: esimatemat}, we can estimate matrix elements using a Hadamard test in \cref{fig:circuit}. For example, 
if we are given a block encoding of $H$, we construct a block encoding of $p(H)$ using the QSVT and estimate matrix elements by setting $U=U_{p(H)}$, a block encoding of $p(H)$, in \cref{fig:circuit} with additional ancilla qubits. This can be similarly done with the QETU. We discuss more technical details in \cref{sec: QSP}. 

Importantly, the use of quantum polynomial filtering allows us to estimate matrix elements of $p(H)$ exponentially efficient than classical matrix element estimation with respect to the degree of polynomial $p(x)$, $\ell$. For instance, if there are $\mathcal{O}(\text{poly}(n))$ Pauli strings in $H$, we see that $p(H)$ is represented by  $\mathcal{O}(\text{poly}(n)^\ell)$ Pauli strings as the worst-case scenario. Then, estimating a matrix element $ \bra{\bm i}p(H)\ket{\bm j}$ requires a $\mathcal{O}(\text{poly}(n)^\ell)$ scaling of computational cost that is spent for estimating the $\mathcal{O}(\text{poly}(n)^\ell)$ coefficients of the Pauli strings in $p(H)$. If we consider to do the same task with matrix-vector multiplications exploiting the sparsity of $H$, we need a similar scaling of computational cost. This task may be done classically in a randomized manner \cite[Appendix B7]{wang2024qubit}, but still a $\mathcal{O}(\text{poly}(n)^\ell)$ scaling of computational cost is required to compute the coefficient weights. Additionally, there is a post-processing that requires a system-dependent classical time for estimating matrix elements at each iteration. However, the exponential dependence of degree $\ell$ is significantly improved as a polynomial dependence and no post-processing is required unlike classical randomized matrix element estimation, if we use quantum computing techniques such as the LCU approach \eqref{eq: esimatemat} or the QSVT and the QETU.  Specifically, in the case of LCU approach, $M=\mathcal{O}(\ell^2)$ if we use the Fe\'{j}er kernel as deduced from \cref{lem: fejerkernel}, and therefore we require $\mathcal{O}(\ell^2)$ query complexity of Hamiltonian simulation for estimating matrix elements of $p(H)$. Moreover, the query complexity is further reduced to a linear scaling $\mathcal{O}(\ell)$ in the case of QSVT \cite{GSLW19} or QETU \cite{dong2022ground} due to their circuit structures that involve a query depth scaling linearly with the degree of polynomial.

\begin{figure}[htbp]
\begin{center}
\text{
\Qcircuit @C=0.9em @R=1.4em {
& \lstick{\ket{0}}  & \gate{\text{H}} &\ctrl{1} & \gate{W} &\gate{\text{H}} & \qw & \meter  & \qw \\
& \lstick{\ket{\bm 0}} & \qw{/}  & \gate{U_i^{\dagger}UU_j} & \qw & \qw &\qw & \qw  & \qw }} 
\end{center}
\caption{ Estimating matrix elements using the Hadamard test given an access to Hamiltonian evolution $e^{iHt}$ or a block encoding of $H$, $U_H$. For \eqref{eq: esimatemat}, $U=e^{iHt}$ and for the QSP techniques \cite{GSLW19,dong2021efficient}, $U$ encodes the matrix polynomial $p(H)$. H is the Hadamard gate and $U_i$ ($U_j$) applies the tensor product of single-qubit X gates corresponding to the computational basis ket $\ket{\bm i}$ ($\ket{\bm j}$). $W$ is either $I$ or $S^{\dagger}$ where $S$ is the phase gate.}
    \label{fig:circuit}
\end{figure}
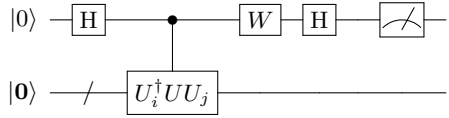 

For our algorithm, we let $p(x)$ be a Chebyshev polynomial, which was motivated by the Chebyshev filtering method \cite{zhou2006self}. 
\begin{figure}[htbp]
    \centering
    \includegraphics[scale=0.3]{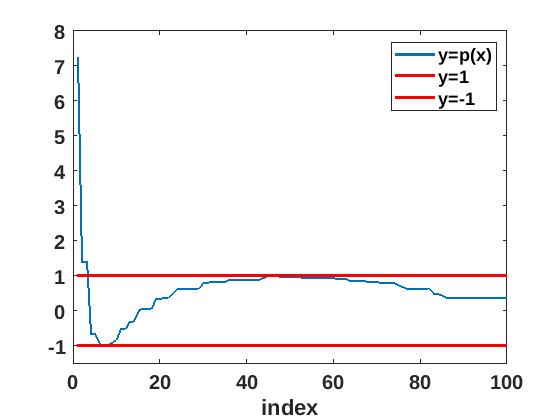}
    \caption{An illustration of the Chebyshev filtering applied to a Hamiltonian matrix $H$. The x-axis labels the index of the eigenvalue of $H$, $\lambda$, in increasing order (e.g. the first index corresponds to the ground-state energy), and the y-axis labels the value of $p(\lambda)$ corresponding to the eigenvalue $\lambda$. }
    \label{fig:cf}
\end{figure}
This method requires a priori bounds to be properly implemented. For example, we need $\lambda_{lb}$ and $\lambda_{ub}$ such that
\begin{equation}\label{bounds}
    \lambda_1<\lambda_{lb}<\lambda_N<\lambda_{ub},
\end{equation}with which the given Hamiltonian is linearly transformed as
\begin{equation}\label{Htilde}
   \frac{2}{\lambda_{ub}-\lambda_{lb}}(H-\lambda_{ub})+I.
\end{equation}Applying this transformation and a Chebyshev polynomial, $p(x)$, to $H$  amplifies the contributions of small eigenvalues but dampens those of large eigenvalues of $H$ as in \cref{fig:cf}, say,
\begin{equation}\label{filtering}
    |p(\lambda_j)|=\begin{cases}
        > 1,\quad \forall j\leq k\\
        < \epsilon,\quad \forall j>k.
    \end{cases}
\end{equation}
Notice that the high-lying eigenvalues in $[\lambda_{lb},\lambda_{ub}]$ are mapped into $[-1,1]$ where the Chebyshev polynomials value in $[-1,1]$, while the low-lying eigenvalues lying in $[\lambda_1,\lambda_{lb})$ are located outside $[-1,1]$. Due to the fact that the Chebyshev polynomial is amplified outside $[-1,1]$,  the contributions of low-lying eigenvalues of $p(H)$ increases rapidly.  

A search of a priori bounds is an important step in the context of filtering algorithms \cite{zhou2006self,lin2020near,lin2022heisenberg}. In classical algorithms, the task can be completed with a combination of Rayleigh quotient computations and Lanczos algorithm as demonstrated in \cite{zhou2006self}. But this may yield an exponential arithmetic operation cost with system size. Alternatively, one may use the binary amplitude estimation with a block encoding of $H$ provided \cite{lin2020near} or with an access to the Hamiltonian evolution $U=e^{iHt}$ \cite{lin2022heisenberg}.

\subsection{Random power method}
Once we have access to the matrix polynomial $p(H)$ as discussed above, we build an iterative algorithm in which a gradient estimator is constructed from the elements of $p(H)$. To do so, we propose an estimator for the matrix-vector product $H\bs x$ by borrowing the ideas of randomized coordinate approach \cite{nesterov2012efficiency,ko2023implementation} and the stochastic gradient descent \cite{bottou2018optimization}. A more precise definition is found in \eqref{unbiasedest}. For simple discussion, we look at the property that the estimator in \eqref{unbiasedest} is proportional to $p(H)\bs x$ in the expectation. Essentially, we define an iterative algorithm as follows, 
\begin{equation}\label{sgd}
    \bs x_{k+1} = \bs x_k-a\bs g_k.
\end{equation}At first glance, this iteration may be irrelevant to the power method. But we here design  $\bs g_k$ such that
\begin{equation}
    \mathbb{E}[\bs g_k]\propto p(H)\bs x_k.
\end{equation}By taking expectation on the equation, we obtain a recursive relation similar to that of power method, 
\begin{equation}\label{eq: gd}
    \bs y_{k+1} \propto (I-ap(H))\bs y_k,\quad y_k:=\mathbb{E}[\bs x_k].  
\end{equation}This can be viewed as a random power method without normalization \cite{lu2020full,shamir2016convergence}. The gradient estimator $\bm g_k$ can be constructed by sampling matrix and vector elements. For example, we sample two indices $i_k,j_k\in [N]$ first, estimate the corresponding matrix element, and define 
\begin{equation}\label{simpleunb}
    \bs g_k = \left((\bm x_k)_{i_k}\bra{\bs i_k}p(H)\ket{\bs j_k}+\epsilon\right)\bs e_{j_k},
\end{equation}where $(\bm x_k)_{i_k}$ is the $i_k$-th component of $x_k$, $\bs e_{j_k}$ is the classical computational basis vector corresponding to index $j_k$. For simplicity, we assume that there is no systematic error and let $\epsilon$ denote the sampling noise. Noting that the two independent random samplings occur for indices $i_k,j_k$, we then observe that averaging $\bm g_k$ over the indices yields the desired matrix-vector product up to a scaling factor,
\begin{equation}\label{eq: mean grad}
    \mathbb{E}[\bs g_k] = \frac{1}{N^2}p(H)\bs x_k.
\end{equation}With this unbiased estimation proportional to the matrix-vector product $p(H)\bm x_k$, we prove  convergence of our algorithm as shown in \cref{thm: convergence2}. 

In the example \eqref{eq: mean grad}, we consider a simple sampling such that the indices are randomly sampled from the set $[N]$, and we obtain the averaging factor $\frac{1}{N^2}$. However, in the proposed algorithm, we assume that we sample the indices corresponding to non-zero elements of $p(H)$ to significantly reduce the required number of iterations of the algorithm. To this end, we require that $p(H)$ is a sparse matrix, and exploit the sparsity of $p(H)$. For example, if $p(H)$ is $s$-sparse and we sample its indices $i_k,j_k$ corresponding to non-zero elements of $p(H)$ by tracking from the indices of non-zero elements of the original Hamiltonian $H$, then we obtain a much better scaling factor than \eqref{eq: mean grad}, 
\begin{equation}
    \mathbb{E}[\bs g_k] = \frac{1}{Ns}p(H)\bs x_k.
\end{equation}

\section{Algorithm description}\label{sec: algorithm}

\subsection*{Gradient estimator}
We extend the estimator in \eqref{simpleunb} to the one defined with more than one matrix element samples.  Let us denote by $\xi(H,i_r,i_c)$, the quantity being an estimate of the element of $H$, which may be obtained from a quantum computer. 
\begin{definition}Given an iterate vector $\bs x$ and a Hamiltonian $H$, we define an estimator $\bs g$ as 
    \begin{equation}\label{unbiasedest}
    \bs g = \sum_{i_c=1}^{m_c}\sum_{i_r=1}^{m_r} x_{i_r}\xi(H,i_r,i_c)\bs e_{i_c}.    
\end{equation} 
Here, $i_r\in [N]$ and $i_c\in[N]$ are distinct indices randomly sampled. Here, $m_r$ and $m_c$ correspond to the numbers of sampled indices $i_r$ and $i_c$, respectively. The coefficient $x_{i_r}$ denotes the $i_r$-th component of $\bs x$. The $\xi(H,i_r,i_c)$ approximates the element of $H[i_r,i_c]$.
\end{definition}
For our applications, the estimate $\xi$ will approximate the element of $p(H)$, such as \eqref{eq: esimatemat}.  For the sake of simplicity, we first consider the case where $\xi(H,i_r,i_c)=\bra{\bm i_r}H\ket{\bm i_c}$.  By the definition \eqref{unbiasedest} and the definition of matrix 2-norm, we have
\begin{equation}\label{eq: ub of g}
    \norm{\bs g}\leq \norm{H}\sqrt{\sum_{i_r=1}^{m_r}x_{i_r}^2}\leq \norm{H}\norm{\bs x}.
\end{equation}

Note that the example \eqref{simpleunb} is a special case of this estimator when $m_r=m_c=1$ without any quantum noise. The following proposition shows that the estimator \eqref{unbiasedest} is proportional to $H\bs x$ in the expectation, and has the variance scaling with $\norm{\bs x}^2$.
\begin{proposition}\label{proposition}
    The estimator defined in \eqref{unbiasedest} satisfies the following properties
\begin{eqnarray}
    & \mathbb{E}[\bs g]=\frac{m_rm_c}{N^2}H\bs x\\
    & \mathbb{E}[\bs g\bs g^*]=\frac{m_c(N-m_c)}{N(N-1)}\mathrm{diag}\left(H\bs x \bs x^*H\right)+\frac{m_c(m_c-1)}{N(N-1)}H\bs x\bs x^*H+\Sigma_{r,c}, 
\end{eqnarray}where the matrix $\Sigma_{r,c}$ defined in \eqref{eq: Sigmarc}
 satisfies that
\begin{equation}
    \mathrm{tr}(\Sigma_{r,c}) \leq \frac{m_c(N-m_r)}{N^2}\norm{H}^2\norm{\bm x}^2.
\end{equation}These properties are improved if $H$ is $s$-sparse and the indices in \eqref{unbiasedest} correspond to non-zero elements of $H$, 
\begin{equation}
     \mathbb{E}[\bs g] =\frac{m_rm_c}{Ns}H\bs x,
\end{equation}and
\begin{equation}
        \mathrm{tr}(\Sigma_{r,c})\leq \frac{m_c(s-m_r)}{Ns}\norm{H}^2\norm{\bm x}^2.
\end{equation}

\end{proposition}
The proof of \cref{proposition} is shown in \cref{sec: proof of proposition} and around \eqref{eq: modified variance} in \cref{sec: proof of convergence thm2}.

We remark that when $\xi$ in \eqref{unbiasedest} estimates the element of $p(H)$, \cref{proposition} can be restated with $p(H)$ instead of $H$.

\subsection*{Algorithm}As discussed in \cref{sec: prem}, we incorporate the estimator \eqref{unbiasedest} to the randomized power method \eqref{sgd}. This outlines the algorithm as shown in \cref{alg: algorithm1} that produces an approximation of a ground state and the corresponding energy estimate.  As a stopping criterion, the Rayleigh quotient,
\begin{equation}\label{RQ}
    \frac{(\bs x_{t},H\bs x_{t})}{\norm{\bs x_{t}}^2},
\end{equation}is computed every iteration. To do so efficiently, we exploit the sparsity of $H$  and $\bs g$ in \eqref{unbiasedest}, otherwise the cost of a brute-force calculation of \eqref{RQ} scales with the dimension $N$. The sparsity of $H$ may inherit from that $H$ is a linear combination of $\mathcal{O}(poly(n))$ many Pauli strings. As a result, the estimation of \eqref{RQ} relies on the following recursive relations 
\begin{subequations}\label{eq: rqrecursive}
    \begin{eqnarray}
    &H\bs x_t = H\bs x_{t-1}-aH\bs g_{t-1}\\
    &\norm{\bs x_{t+1}}^2=\norm{\bs x_t}^2-a(\bs g_t,\bs x_t)-a(\bs x_t,\bs g_t)+a^2\norm{\bs g_t}^2\\
    &(\bm x_{t+1},H\bs x_{t+1})=(\bs x_t,H\bs x_t)-a(\bm g_t,H\bs x_t)-a(\bs x_t,H\bs g_t)+a^2(\bm g_t,H\bs g_t).
\end{eqnarray}
\end{subequations}
Notice that the first relation requires $\mathcal{O}(m_rm_c)$ cost, since $\bs g_{t-1}$ is $m_c$-sparse and the $m_r$-sparse columns of $H$ corresponding to $\bm g_{t-1}$ are added to the vector $H\bs x_{t-1}$. Similarly, the second relation needs $\mathcal{O}(m_c)$ cost. The last relation takes $\mathcal{O}(m_c\min\{m_r,m_c\})$ cost as the last term on the right side dominates overall cost. Therefore, the per-iteration cost for the estimation of \eqref{RQ} scales only as $\mathcal{O}(m_rm_c)$.

\begin{algorithm}
\SetAlgoLined
	\KwData{initial guess $\bm  x_0$, $H\bs x_0$,  step size  $a\in(0,1)$, sparsity parameter $m_r$, index parameter $m_c$ }
	\KwResult{approximate ground-state and energy }

	\For{$t=0:T$}{
             Pick $m_r, m_c\in[N]$ different indices randomly\;

             Construct $\bs g_t$ in \eqref{unbiasedest}
             based on the circuit in  \cref{fig:circuit} with repetition\; 

            $H\bs x_t = H\bs x_{t-1}-aH\bs g_{t-1}$\;

            Compute the quantities in \eqref{eq: rqrecursive}\;
              }
	\caption{Quantum random power method}
    \label{alg: algorithm1}
\end{algorithm}

\section*{Rigorous guarantee}

In this section, we provide theoretical results for \cref{alg: algorithm1}. Our main result can be summarized as follows,
\begin{theorem}[Implication from \cref{thm: convergence2} and \cref{lem: overlap1}]\label{thm: main result}
    For any precision $\epsilon>0$, there exists a quantum algorithm that solves \eqref{problem} by outputing an approximation of a ground state of $H\in\mathbb{C}^{N\times N}$ up to fidelity $1-\epsilon$ with probability one. Especially when $p(H)$ is $s$-sparse, the number of iterations scales as $\mathcal{O}\left(\frac{Ns^2}{\epsilon^2}\log\frac{1}{\abs{\bra{\psi_{GS}}\ket{\bs x_1}}}\right)$. The query depth of the circuit \cref{fig:circuit} for the algorithm is 
$\mathcal{O}(1)$ if the LCU approach is used, and $\mathcal{O}(\frac{1}{\sqrt{\Delta}})$ if the QSVT or the QETU  is used. 

\end{theorem}
This theorem combines \cref{thm: convergence2} and \cref{lem: overlap1}. \cref{thm: convergence2} tells us that with probability one, \cref{alg: algorithm1} converges to a ground state of a matrix that is constructed from the matrix-element estimation $\xi$ in \eqref{unbiasedest}. Since $\xi$ may involve noise occuring in quantum computation, the ground state, to which the iteration converges, may be a perturbation of the true one. Thus, one needs to quantify the resulting fidelity, probably not one, influenced by the noise. Related to this, \cref{lem: overlap1} makes a connection between the fidelity and the magnitude of noise under a small-noise assumption.

\subsection*{Stochastic gradient sampling}
We note that in \cref{alg: algorithm1}, the gradient estimate \eqref{unbiasedest} may involve bias due to the use of the quantum filtering technique \eqref{filtering} and the sampling with repetition. This means that the matrix-element estimation is performed not for $p(H)$, but for a perturbed one $\widetilde{p(H)}$. As a consequence, \cref{alg: algorithm1} converges to an approximation of  ground state of $H$ as in the following theorem. In the next section, we provide a lower bound of the fidelity between perturbed and true ground states to ensure that the more accurate the estimator \eqref{unbiasedest} is computed, the higher the achievable fidelity of \cref{alg: algorithm1} is, as being arbitrarily close to one. 

The following theorem guarantees convergence with probability one to a given precision if a sufficiently small step size is chosen. As usually observed in variants of the power method \cite{lu2020full,shamir2016convergence}, the iteration complexity has a logarithmic dependence on the overlap of the ground state. By combining \cref{thm:complexity}, we present the overall query complexity for \cref{alg: algorithm1}. 
\begin{theorem}\label{thm: convergence2}
    Assume that $p(H)$ is $s$-sparse. For any precision $\epsilon>0$, there exists a step size $a>0$ and $T=\mathcal{O}\left(\frac{Ns^2}{\epsilon^2}\log\frac{1}{\abs{\bra{\psi_{GS}}\ket{\bs x_1}}}\right)$ such that 
     \cref{alg: algorithm1} outputs an iterate whose fidelity is greater than $1-\epsilon$ with probability one, namely,  
    \begin{equation}
     \mathbb{P}\left(\left[r_T\geq 1-\epsilon, r_t<1-\epsilon\text{ for all }t<T \right]\text{ or }\left[r_t\geq 1-\epsilon\text{ for some }t<T\right]\right)=1,
    \end{equation}where $r_t$ denotes the squared fidelity with the subspace of degenerate ground states, namely,
    \begin{equation}
        r_t:=\frac{\sum_{j=1}^g\abs{(\psi_{GS,j},\bs x_{t})}^2}{\|\bs x_{t}\|^2}.
    \end{equation}Here, the $\psi_{GS,j}$'s are the ground states of $\widetilde{P(H)}$ whose 
$(i_r,i_c)$-element is determined by the estimation  $\xi(H,i_r,i_c)$  \eqref{unbiasedest} for $p(H)$.
\end{theorem}
The proof of \cref{thm: convergence2} is shown in \cref{sec: proof of convergence thm2}.

\subsection*{A lower bound of the fidelity between a perturbed ground state and the subspace of true ones}\label{sec: proof of projector thm}

In \cref{alg: algorithm1}, the elements of a matrix polynomial could be estimated without repetition. The resulting matrix, $\widetilde{p(H)}$, will be one perturbed from the desired matrix $p(H)$ by at least the sampling noise, denoted as $E$, namely,
\begin{equation}\label{error}
    \norm{E}=\norm{\widetilde{p(H)}-p(H)}>0.
\end{equation}Here, we denote an error-free matrix and noisy one by $A$ and $\widetilde{A}$. Suppose that these matrices are Hermitian as in the situation \eqref{error}. It may be challenging to derive a perturbative relation between the ground states of $A$ and $\widetilde{A}$ when the error $E$ is considered to be arbitrary. But when $E$ is smaller than some value, we can show that a lower bound for the fidelity between a perturbed ground state and the true ones depends only on the magnitude of noise, $\norm{E}$. 

Let $\lambda_1(A)$ and $\lambda_2(A)$ be the first and second smallest eigenvalues of $A$, and similar for $\widetilde{A}$.  Let $E=\widetilde{A}-A$ denote the error. The following lemma shows that when $\norm{E}$ is smaller than a certain value, the overlap of any ground state of $\widetilde{A}$ with the subspace of ground states of $A$ increases to $1$ linearly with $\norm{E}$. 
\begin{lemma}\label{lem: overlap1}
 Assume that $\norm{E}\leq e$ for some $e\in(0,\frac{1}{2}\abs{\lambda_1(A)-\lambda_2(A)})$. Then, for any precision $\epsilon>0$, there exists a constant $C(\norm{A},e,\epsilon)>0$ such that \begin{equation}
        \bra{\widetilde{ \psi}}P_{\lambda_1(A)}\ket{\widetilde{ \psi}}\geq 1-C(\norm{A},e,\epsilon)\norm{E}-\epsilon,
    \end{equation}where $\ket{\widetilde{ \psi}}$ is any ground state of $\widetilde{A}$, $P_{\lambda_1(A)}$ denotes the projection of $A$ onto its subspace of ground states, and the constant $C(\norm{A},e,\epsilon)$ depends on $\norm{A}$, $e$ and $\epsilon$.
\end{lemma}The proof of this lemma is shown in \cref{sec: proof of overlap1}.

\section{Complexity of constructing gradient estimator}\label{sec: sec of circuit complexity}

We discuss about the overall quantum complexity required to construct the gradient estimator $\bm g_t$ in \cref{alg: algorithm1}. As described in \cref{sec: matrix-element estimation}, the gradient estimator \eqref{unbiasedest} is constructed with either of the LCU approach using a Fourier-series approximation or the QSP techniques such as QSVT and QETU.

The estimation of $ \bra{\bm i}p(H)\ket{\bm j}$ in \eqref{eq: esimatemat} at each iteration of our algorithm requires $\mathcal{O}(M)$ values of $\bra{\bm i}e^{iHm}\ket{\bm j}$ for integer $m\in[-M,M]$, each of which requires Trotter steps as many as $\mathcal{O}(\text{poly}(n)m^2)$ if the first-order trotter formula is considered \cite{childs2019faster}. Here, a Chebyshev polynomial filtering in \cref{alg: algorithm1} is used whose degree scales as $\mathcal{O}(\frac{1}{\sqrt{\Delta}})$ and $M=\mathcal{O}(\frac{1}{\Delta^2})$ as discussed in \cref{sec: fejer}. That is, sampling matrix elements in \cref{alg: algorithm1} using Hamiltonian simulation requires at most $\mathcal{O}(\text{poly}(n)\frac{1}{\Delta^4})$ gate complexity with $n +1$ qubits and is performed  via the Hadamard test in \cref{fig:circuit}. In addition, $\mathcal{O}(\frac{1}{\Delta^2})$ circuit sampling is required as the $m$ iterates over $[-M,M]$, given that a fixed number of circuit sampling is set for each $m$. We remark that the dependence of gate complexity on the spectral gap could be improved if a higher-order trotter formula is used, as it depends on the maximal time $M$.  

\cref{alg: algorithm1} using the QSVT \cite{GSLW19} (or the QETU \cite{dong2021efficient}) requires   $\mathcal{O}(\frac{1}{\sqrt{\Delta}})$ applications of the block encoding of $H$ and its inverse (or the Hamiltonian simulation), since the degree of Chebyshev filtering polynomial scales as $\mathcal{O}(\frac{1}{\sqrt{\Delta}})$.  In our numerical simulations in \cref{sec: numercial result}, we considered typical physical quantum models whose Hamiltonians are expressed as a linear combination of Pauli strings. In this case, the required circuit complexity for constructing a block encoding of $H$ scales polylogarithmically with system size \cite[Theorem 7]{zhang2024circuit}. Therefore, the overall circuit complexity for computing $ \bra{\bm i}p(H)\ket{\bm j}$ using QSVT is $\mathcal{O}(\frac{1}{\sqrt{\Delta}} \text{poly}(n))$. It is the same for the case of QETU since the trotter formula for each query of Hamiltonian simulation in the $\mathcal{O}(\frac{1}{\sqrt{\Delta}})$ query depth requires $\mathcal{O}(\text{poly}(n))$ gate complexity.

From either of the two approaches, we estimate a fixed number of matrix elements of the form  $ \bra{\bm i}p(H)\ket{\bm j}$, and construct the $m_c$-sparse gradient estimator, $\bm g_t$, at each iteration in \cref{alg: algorithm1}. Thus, the per-iteration classical time and memory cost for updating the iteration are independent of system size as scaling as $\mathcal{O}(1)$.



\section{Numerical result}\label{sec: numercial result}
In this section, we present several numerical experiments to demonstrate the performance of our algorithm. For the Hamiltonian simulation approach, we work with the transverse-field Ising model (TFIM) and the XXZ model, and examine the performance of the algorithm with and without polynomial filtering methods. For the block-encoding approach, we provide numerical tests that include the TFIM and one-dimensional Hubbard model (HM). We compare convergence results that are obtained using \cref{alg: algorithm1} with and without polynomial filters. Additionally, we discuss how the filtering methods alleviate the effect of noise in the two approaches.   

We set up the TFIM as
\begin{equation}\label{eq:isingmodel}
    H = J\sum_{j=1}^{n-1}Z_jZ_{j+1}+D\sum_{j=1}^nX_j.
\end{equation}Here, $n$ denotes the number of qubits, and $X_j,Z_j$ are Pauli operators acting on the $j$-th qubit. The  XXZ model is defined as \begin{equation}\label{eq:xxzmodel}
    H = J\sum_{j=1}^{n-1}(X_jX_{j+1}+Y_jY_{j+1})-D\sum_{j=1}^{n-1}Z_jZ_{j+1}.
\end{equation} The HM, which is defined on a one-dimensional lattice with open boundary condition, is expressed as  
\begin{equation}
    H = -t\sum_{\sigma\in\{\uparrow,\downarrow\}}\sum_{j=1}^{n-1}\left(a_{j,\sigma}^{\dagger}a_{j+1,\sigma}+a_{j+1,\sigma}^{\dagger}a_{j,\sigma}\right)+U\sum_{j=1}^n\hat{n}_{j,\uparrow}\hat{n}_{j,\downarrow},
\end{equation}where the parameters $t$ and $U$ account for the hopping energy and the Coulomb energy.
After applying a Jordan-Wigner transformation \cite{mcardle2020quantum}, we obtain 
\begin{equation}\label{eq:hmodel}
    H = -\frac{t}{2}\sum_{\sigma\in\{\uparrow,\downarrow\}}\sum_{j=1}^{n-1}\left(X_{j+1,\sigma}X_{j,\sigma}+Y_{j+1,\sigma}Y_{j,\sigma}\right)+\frac{U}{4}\sum_{j=1}^n\left(-Z_{j,\uparrow}-Z_{j,\downarrow}+Z_{j,\uparrow}Z_{j,\downarrow}\right).
\end{equation}For convenience, a multiple of the identity matrix in \eqref{eq:hmodel} is omitted, since it does not affect the optimization.

\subsection{Experiment detail}
For the Hamiltonian simulation approach (i.e. \eqref{eq:isingmodel} and 
\eqref{eq:xxzmodel}), we simulate \cref{alg: algorithm1} using Qiskit \cite{Qiskit}. To estimate the gradient, we incorporate the Chebyshev filtering method using the first-order Trotter formula, the Fourier-series approximation with the Fe\'{j}er kernel, and Hamiltonian simulation. Thus three types of errors are involved: numerical errors due to the Fourier expansion and the Trotter formula as well as the sampling noise from the Hadamard test.  For the block-encoding approach (i.e. \eqref{eq:isingmodel} and \eqref{eq:hmodel}), we simulate the quantum singular value transformation (QSVT) on MATLAB by computing the eigen-decomposition of a given Hamiltonian and constructing a corresponding matrix polynomial to mimic a block encoding of a matrix polynomial of $H$. In this approach, only sampling noise is considered. In all numerical experiments, we save the estimated matrix elements once and use the estimated matrix to carry out 100 independent simulations for \cref{alg: algorithm1} with the same initial guess.

\subsection{LCU approach (title changed)}\label{sec: Ham approach}
We consider the models \eqref{eq:isingmodel} and \eqref{eq:xxzmodel} for different $D/J$ with $J=1$.  For every $D/J$, we set the same parameter values, $n=2$, and the shot number to be $4\times 10^5$, the degree of the Fe\'{j}er kernel to be $30$, the degree of the Chebyshev filtering to be $3$, and $m_r=m_c=1$ in \cref{alg: algorithm1}. We use the first-order Trotter formula and set $\max\{k^2+20,50\}$ to be the number of Trotter steps when estimating matrix elements of $e^{iHk}$ where $k$ is a positive integer. In \cref{fig:TFIMXXZ} that compares the performance of \cref{alg: algorithm1} with and without polynomial filters, we measure the growing rate of the average fidelity for each $D/J$, namely,  
\begin{equation}\label{eq: growing rate}
    \frac{A}{A(1-\frac{\epsilon}{2})+B\frac{\epsilon}{2}},
\end{equation}where the constants $A$ and $B$ are
\begin{equation}
    A = \left(1-\frac{a\lambda_1}{N^2}\right)^2,\quad B = 1-\frac{2a\lambda_2}{N^2}+\frac{a^2}{N}\max_{j\neq 1}\lambda_j^2.
\end{equation}Note that this quantity measures how fast the average fidelity increases, according to the technical detail \eqref{eq: coeffs3}. In the following numerical results, $m_r=m_c=1$. Here we set $\epsilon=0.9$. Under this setting, the quantity \eqref{eq: growing rate} tells us the worst-case growing rate of the average fidelity unitl it reaches the precision $\epsilon$.  For fair comparison, we compare the cases without scheduling decreasing step sizes in the bottom panels of \cref{fig:TFIMXXZ}, as the step size $a$ is a constant. For the polynomial filtering, we set the upper and lower bounds in \eqref{bounds} as $\lambda_{ub}=\lambda_{\max}+0.2(\lambda_{\max}-\lambda_{\min})$ and $\lambda_{lb}=\lambda_{\min}+0.2(\lambda_{\max}-\lambda_{\min})$ for the TFIM and the XXZ model.

The top panels in \cref{fig:TFIMXXZ} show the convergence of our algorithm with and without polynomial filters, and with and without scheduling decreasing step sizes. The graphs show the average of fidelities obtained from the 100th iterates of 100 independent simulations with a fixed initial guess for each $D/J$. For the two models, we can see that in all cases, the optimization with the filtering  method converges much faster than the ones without it. One reason may be that the convergence rate of the fidelity is related to the growing rate of the average fidelity \eqref{eq: growing rate}. We observe that the overall growing rate is relatively higher when the filtering is applied as shown in the bottom panels of \cref{fig:TFIMXXZ}. In addition, the convergence result in \cref{fig:TFIMXXZ} shows that \cref{alg: algorithm1} with the filtering method approximates the true ground state well, despite the additional errors involved (e.g. the Trotter error, the Fourier approximation error, and the sampling error). Together with the convergence guarantee in \cref{thm: convergence2}, this can be verified by checking the high fidelities in the last column of \cref{tab: filtering}, which are the fidelity of the ground state from the errorneous matrix contructed from \cref{alg: algorithm1} with the true ground state. 

Regarding \cref{lem: overlap1}, we observe that the perturbation error (LHS), $\norm{E}$, is smaller than half of the enlarged spectral gap (RHS), $\frac{1}{2}\abs{\lambda_1(p(H))-\lambda_2(p(H))}$, in all cases. This means that the condition in \cref{lem: overlap1} is satisfied in the cases, and so is the lower bound in the lemma. Accordingly, the ordering of the first two smallest eigenvalues is preserved even in the presence of sampling noise, and moreover the interval $[\lambda_1-e,\lambda_1+e]$ filters out the smallest eigenvalue of the estimated matrix from the second smallest one.


\begin{figure}[htbp]

    \centering
    \begin{subfigure}[b]{0.4\textwidth}
    \includegraphics[width=\textwidth]{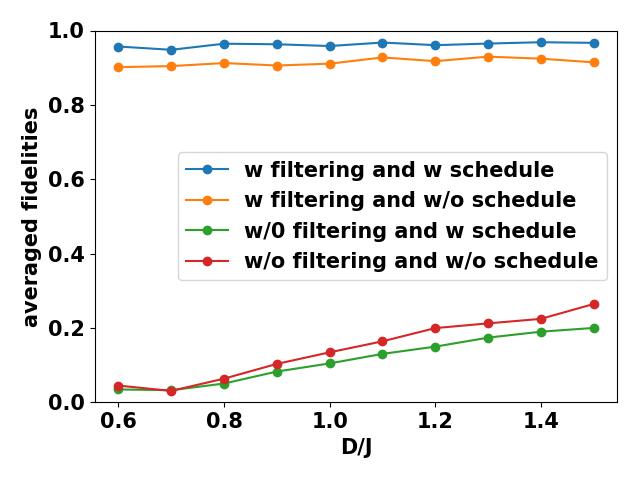}
    \end{subfigure}
    \begin{subfigure}[b]{0.4\textwidth}
    \includegraphics[width=\textwidth]{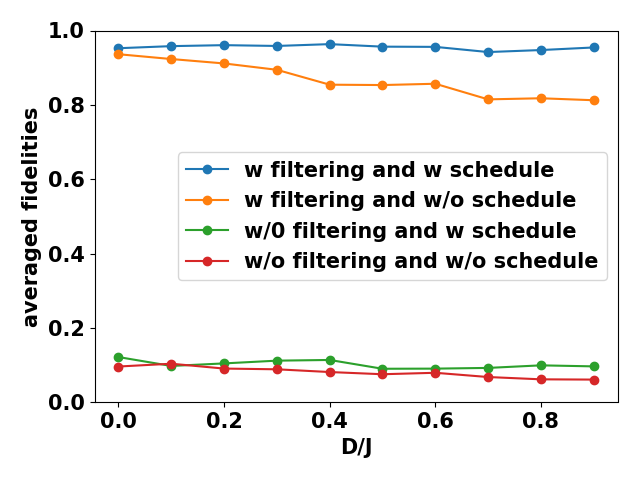}
    \end{subfigure}

    \begin{subfigure}[b]{0.4\textwidth}
    \includegraphics[width=\textwidth]{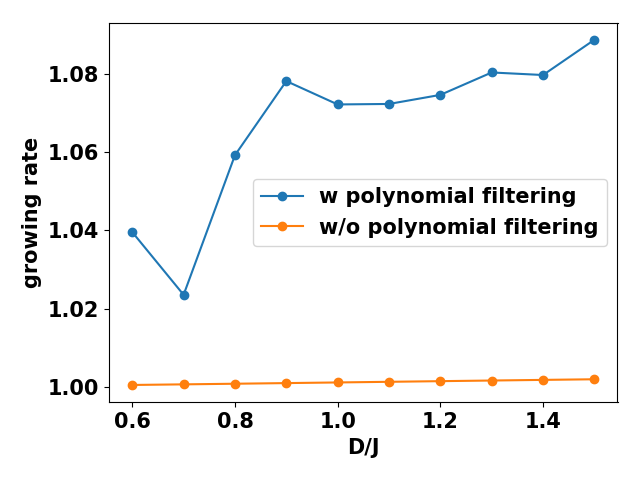}
    \end{subfigure}
    \begin{subfigure}[b]{0.4\textwidth}
    \includegraphics[width=\textwidth]{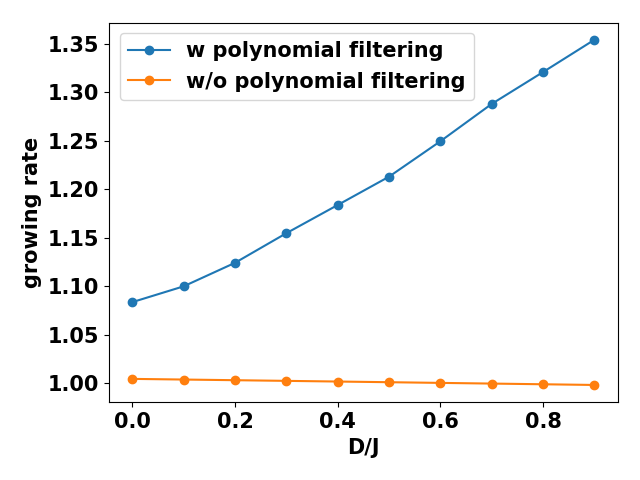}
    \end{subfigure}
    \caption{\textbf{Top}: Comparison of the averaged fidelities obtained after 100 iterations of \cref{alg: algorithm1} with and without the polynomial filter, and with and without scheduling decaying step sizes for the TFIM (left) and the XXZ model (right) for different $D/J$. For each model, we used the same initial guess. For the optimization with step-size scheduling, we decrease the step size by a rate of 0.5 every 10 iteration steps. \textbf{Bottom}: Comparison of the growing rate of average fidelity for the cases with and without polynomial filtering.}
    \label{fig:TFIMXXZ}
\end{figure}


\begin{table}
    \centering
    \begin{subtable}[h]{0.8\textwidth}
        \centering
         \begin{tabular}{c|c|c|c|c|c|c|c}
         $D/J$ & LHS & RHS & $\widetilde{\lambda_1}$  & $\widetilde{\lambda_2}$  & $\lambda_1-\text{LHS}$ & $\lambda_1+\text{LHS}$ & $\abs{\bra{\widetilde{\psi}}\ket{\psi_{true}}}$\\
        \hline
         0.5 & 1.36 & 2.37 & -9.58  & -4.45 & -10.94 & -8.21 & 0.99\\
         0.6 & 1.17 & 2.86 & -9.75 & -2.84 & -10.93 & -8.57 & 0.99\\
         0.7 & 1.02 & 3.25 & -8.73 & -3.02 & -9.75 & -7.71 & 0.99\\
         0.8 & 0.49 & 3.54 & -9.16 & -1.96 & -9.65 & -8.67 & 0.99\\
         0.9 & 0.85 & 3.74 & -9.03 & -1.09 & -9.88 & -8.18 & 0.99\\
         1.0 & 0.87 & 3.88 & -8.88 & -1.30 & -9.76 & -8.01 & 0.99\\
         1.1 & 0.77 & 3.97 & -8.31 & -1.02 & -9.09 & -7.53 & 0.99 \\
         1.2 & 1.02 & 4.03 & -8.61 & -1.05 & -9.63 & -7.59 & 0.99\\
         1.3 & 0.5 & 4.07 & -8.76 & -0.83 & -9.26 & -8.26 & 0.99\\
         1.4 & 0.61 & 4.09 & -8.04 & -0.42 & -8.65 & -7.42 & 0.99\\
         1.5 & 0.78 & 4.11 & -9.21 & -0.59 & -10.00 & -8.42 & 0.99 \\
    \end{tabular}
    \caption{The TFIM}
    \end{subtable}
    
   \begin{subtable}[h]{0.8\textwidth}
        \centering
         \begin{tabular}{c|c|c|c|c|c|c|c}
          $D/J$ & LHS & RHS & $\widetilde{\lambda_1}$  & $\widetilde{\lambda_2}$  & $\lambda_1-\text{LHS}$ & $\lambda_1+\text{LHS}$ & $\abs{\bra{\widetilde{\psi}}\ket{\psi_{true}}}$\\
        \hline
         -0.1 & 0.49 & 3.28 & -6.37 & 0.09 & -6.86 & -5.88 & 0.99  \\
         0.0 & 0.5 & 4.07 & -8.03 & 0.08 & -8.53 & -7.53 & 0.99 \\
         0.1 & 0.54 & 5.03 & -10.16 & -0.28 & -10.70 & -9.61 & 0.99  \\
         0.2 & 0.29 & 6.39 & -13.18 & -0.72 & -13.48 & -12.88 & 0.99 \\
         0.3 & 0.61 & 8.03 & -16.45 & -0.37 & -17.07 & -15.83 & 0.99  \\
         0.4 & 0.76 & 9.72 & -19.84 & -0.22 & -20.61 & -19.08 & 0.99  \\
         0.5 & 0.51 & 11.62 & -23.57 & -0.61 & -24.09 & -23.05 & 0.99  \\
         0.6 & 0.75 & 13.96 & -27.77 & -0.17 & -28.52 & -27.01 & 0.99  \\
         0.7 & 0.77 & 16.6 & -33.46 & -0.27 & -34.24 & -32.69 & 0.99  \\
         0.8 & 0.7 & 19.35 & -38.72 & -0.39 & -39.43 & -38.01 & 0.99  \\
         0.9 & 0.74 & 22.33 & -44.27 & -0.10 & -45.02 & -43.52 & 0.99  \\
    \end{tabular}
    \caption{The XXZ model}
    \end{subtable}
    \caption{For each $D/J$, the second and third columns show the spectral norm of the noise $\norm{E}$ and a half of the spectral gap of the noise-free matrix $\frac{\abs{\lambda_1-\lambda_2}}{2}$ in \cref{lem: overlap1}. The fourth and fifth columns show the two smallest eigenvalues, which are denoted by $\widetilde{\lambda}_1$ and $\widetilde{\lambda}_2$, of the perturbed matrix polynomial $\widetilde{p(H)}$. The values in the sixth and seventh columns form an interval $[\lambda_1-\norm{E},\lambda_1+\norm{E}]$ that may contain $\widetilde{\lambda_1}$ but $\widetilde{\lambda_2}.$ The last column displays the absolute values of fidelity between the perturbed ground state of $\widetilde{p(H)}$ and the true one of $H$.}
    \label{tab: filtering}
\end{table}

\subsection{QSVT approach (title changed)}
In this section, we consider the TFIM \eqref{eq:isingmodel} with $J=1, D=1.5$ and the Hubbard model \eqref{eq:hmodel} in weak ($t=1,U=1$) and strong coupling regimes ($t=1,U=6$).  we set $n=10$, the parameters $m_r=20,m_c=20$ in \cref{alg: algorithm1}, and the degree of the Chebyshev polynomial to be $7$ for the TFIM and $8$ for the Hubbard model. We note that there is a non-degenerate ground state for the cases of \eqref{eq:isingmodel} and \eqref{eq:hmodel} in the weak-coupling regime, while there are two degenerate ground states in the strong-coupling regime for \eqref{eq:hmodel}. For this case, we measure the fidelity as follows,
\begin{equation}\label{totalovp}
    \frac{\sqrt{\abs{(\psi_{GS,1},\bs x_t)}^2+\abs{(\psi_{GS,2},\bs x_t)}^2}}{\norm{\bm x_t}},
\end{equation}where the $\psi_{GS,i}$'s for $i=1,2$ denote two orthogonal ground states. This is a special case of more general result in \cref{sec: proof of projector thm}, which considers the number of ground states to be $2$.  For the polynomial filtering, we set the upper and lower bounds in \eqref{bounds} as $\lambda_{ub}=\lambda_{\max}+0.1(\lambda_{\max}-\lambda_{\min})$ and $\lambda_{lb}=\lambda_{\min}+0.03(\lambda_{\max}-\lambda_{\min})$ for the TFIM and the Hubbard model in the weak-coupling regime, and the same upperbound and $\lambda_{lb}=\lambda_{\min}+0.02(\lambda_{\max}-\lambda_{\min})$ for the HM in the strong-coupling regime. These lower bounds, $\lambda_{lb}$'s, seperate $\lambda_1$ and $\lambda_2$ for the TFIM model, $\lambda_3$ and $\lambda_4$ for the HM in the weak-coupling regime, and $\lambda_{15}$ and $\lambda_{16}$ for the HM in the strong-coupling regime, respectively.

\cref{fig:ex1} shows the effect of the polynomial filter and the shot number on the convergence of \cref{alg: algorithm1}. One can observe that  the number of shots influences the performance of \cref{alg: algorithm1}. For example, the result with the largest shot number is much better than the others. Additionally, from the bottom panels of \cref{fig:ex1}, it is evident that the use of the polynomial filter improves overall performance of \cref{alg: algorithm1}, exhibiting better convergence  than the result without it.  Importantly, the result in \cref{fig:ex1} shows high fidelities, which means that \cref{alg: algorithm1} with the filtering method provides a fairly good approximation even with the sampling noise. By the convergence theorem \cref{thm: convergence2}, this can be understood by looking at the fidelity of approximate ground states with the true ones as in the last columns in \cref{tab: filtering1}. 

As discussed in the previous section, we compare the perturbation error (LHS) to the enlarged spectral gap (RHS). For each of the three models, the condition in \cref{lem: overlap1} is not satisfied in some cases where the LHS is larger than the RHS, but the interval $[\lambda_1-e,\lambda_1+e]$ filters the smallest perturbed eigenvalue $\widetilde{\lambda_1}$ from the second one, which means that the ordering of the first two smallest eigenvalues is preserved even in the presence of sampling noise. However, this observation implies that the condition in \cref{lem: overlap1} can be restrictive in practice. It would be interesting to find a more relaxed condition, and we leave this to future work.

For the Hubbard model in a strong coupling regime, among 100 simulation results, \cref{fig:ex2} in \cref{sec: additional results} plots 12 that display the respective fidelities with the two degenerate ground states and the mean-square sum with normalization as defined in \eqref{eq: fidelity}. The graphs mean that our method converges to a ground state as shown in \cref{thm: convergence2}.

\begin{figure}[htbp]
    \centering
    \begin{subfigure}[b]{0.3\textwidth}
    \includegraphics[width=\textwidth]{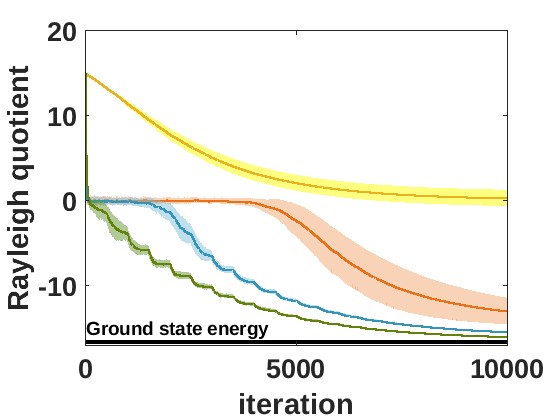}
    \end{subfigure}
    \begin{subfigure}[b]{0.3\textwidth}
    \includegraphics[width=\textwidth]{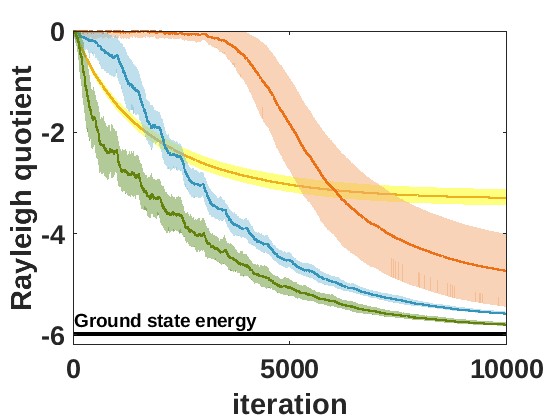}
    \end{subfigure}
    \begin{subfigure}[b]{0.3\textwidth}
    \includegraphics[width=\textwidth]{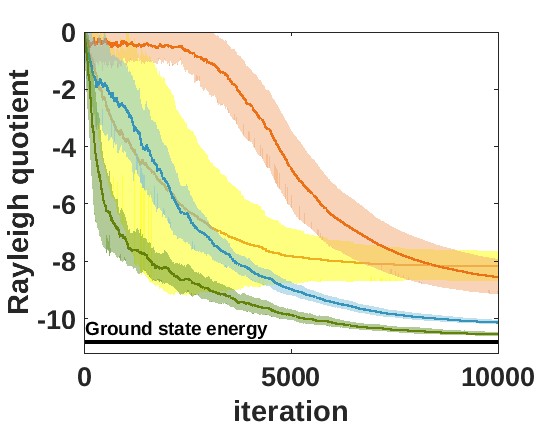}
    \end{subfigure}

    \begin{subfigure}[b]{0.3\textwidth}
    \includegraphics[width=\textwidth]{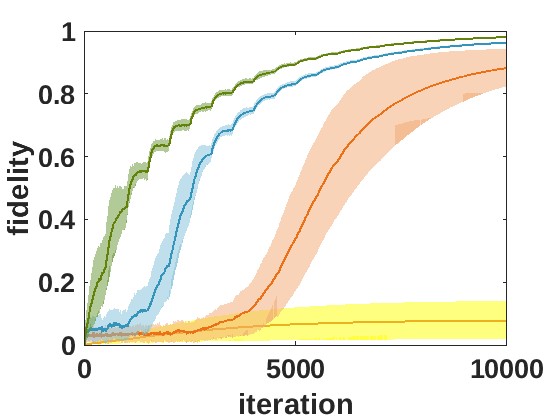}
    \end{subfigure}
    \begin{subfigure}[b]{0.3\textwidth}
    \includegraphics[width=\textwidth]{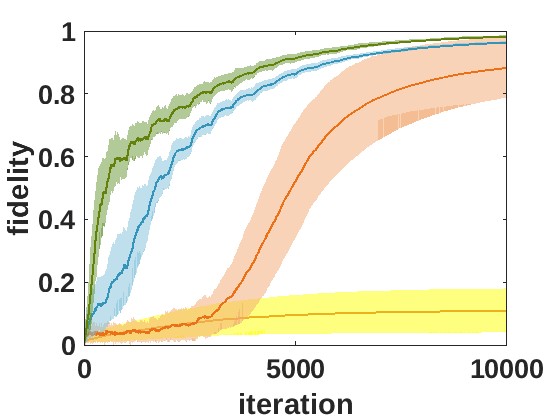}
    \end{subfigure}
    \begin{subfigure}[b]{0.3\textwidth}
    \includegraphics[width=\textwidth]{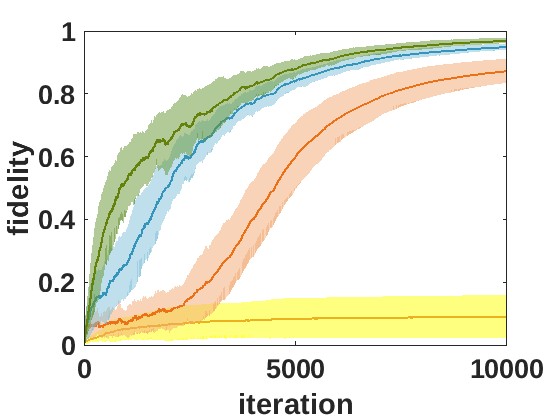}
    \end{subfigure}
    \caption{\textbf{Top}: comparison of the Rayleigh quotients obtained from \cref{alg: algorithm1} when using the polynomial filter with different shot numbers ($n_{shot}=10^5(\text{orange}),\; 4\times 10^5 (\text{blue}),\; 4\times 10^6 (\text{green})$) and without the filter (yellow) for the TFIM (left), the HM in a weak coupling regime (middle), and the HM in a strong regime (right). In all optimization results, we decrease the step size by a rate of 0.8 every 500 iteration steps. \textbf{Bottom}: comparison of the fidelity of the iteration with the true ground state or the subspace of true ground states. For each of the three models, we use a fixed initial guess.} 
    \label{fig:ex1}
\end{figure}

\begin{table}
    \centering
    \begin{subtable}[h]{0.8\textwidth}
        \centering
         \begin{tabular}{c|c|c|c|c|c|c|c}
         $n_{shot}$ & LHS & RHS &  $\widetilde{\lambda_1}$  & $\widetilde{\lambda_2}$  & $\lambda_1-\text{LHS}$ & $\lambda_1+\text{LHS}$ & $\abs{\bra{\widetilde{\psi}}\ket{\psi_{true}}}$\\
        \hline
         $10^5$ & 2.11 & 2.10 & -5.41 & -2.46 &-7.31 & -3.09 &0.95\\
         $4\times 10^5$ & 1.05 & 2.10 & -5.25 & -1.59 & -6.25 & -4.15 & 0.98\\
         $4\times 10^6$ & 0.33 & 2.10 & -5.21 & -1.12 & -5.53 & -4.87 & 0.99\\
    \end{tabular}
    \caption{The TFIM}
    \end{subtable}

    \begin{subtable}[h]{0.8\textwidth}
        \centering
         \begin{tabular}{c|c|c|c|c|c|c|c}
         $n_{shot}$ & LHS & RHS &  $\widetilde{\lambda_1}$  & $\widetilde{\lambda_2}$  & $\lambda_1-\text{LHS}$ & $\lambda_1+\text{LHS}$ & $\abs{\bra{\widetilde{\psi}}\ket{\psi_{true}}}$\\
        \hline
         $10^5$ & 3.09 & 2.91 & -7.62 & -3.37 & -10.32 & -4.13 &  0.95\\
         $4\times 10^5$ & 1.54 & 2.91 & -7.29 & -2.02 & -8.77 & -5.68 & 0.98\\
         $4\times 10^6$ & 0.48 & 2.91 & -7.24 & -1.44 & -7.72 & -6.74 & 0.99\\
    \end{tabular}
    \caption{The HBM weak}
    \end{subtable}
    
    \begin{subtable}[h]{0.8\textwidth}
        \centering
         \begin{tabular}{c|c|c|c|c|c|c|c|c}
         $n_{shot}$ & LHS & RHS & $\widetilde{\lambda_1}$  & $\widetilde{\lambda_2}$  & $\lambda_1-\text{LHS}$ & $\lambda_1+\text{LHS}$ & $\bra{\psi_{B,1}}P_1\ket{\psi_{B,1}}$ &  $\bra{\psi_{B,2}}P_1\ket{\psi_{B,2}}$ \\
        \hline
         $10^5$ & 2.04 & 0.93 & -4.72 & -3.01 & -6.48 & -2.39 & 0.95 & 0.94\\
         $4\times 10^5$ & 1.01 & 0.93 & -4.53 & -2.71 & -5.45 & -3.42 & 0.98 & 0.98\\
         $4\times 10^6$ & 0.32 & 0.93 & -4.45 & -2.58 & -4.76 & -4.11 & 0.99 & 0.99\\
    \end{tabular}
    \caption{The HBM strong}
    \end{subtable}
    \caption{For each $n_{shot}$, the second and third columns show the spectral norm of the noise $\norm{E}$ and a half of the spectral gap of the noise-free matrix $\frac{\abs{\lambda_1-\lambda_2}}{2}$ in \cref{lem: overlap1}. The fourth and fifth columns show the two smallest eigenvalues, which are denoted by $\widetilde{\lambda}_1$ and $\widetilde{\lambda}_2$, of the perturbed matrix polynomial $\widetilde{p(H)}$. The values in the sixth and seventh columns form an interval $[\lambda_1-\norm{E},\lambda_1+\norm{E}$ that may contain $\widetilde{\lambda_1}$ but $\widetilde{\lambda_2}.$ The last column displays the absolute values of fidelity between the perturbed ground state of $\widetilde{p(H)}$ and the true one of $H$ for the TFIM and the Hubbard model in a weak-coupling regime. For the Hubbard in a strong-coupling regime, we measure the fidelity of the two perturbed ground states of $\widetilde{p(H)}$ with the subspace of true ones of $H$, on which the orthogonal projection $P_1$ is defined.} 
    \label{tab: filtering1}
\end{table}


\subsection{The sparsity of $p(H)$}\label{subsec: sparsity}

The sparsity value of $p(H)$, $s$, depends on the Hamiltonian matrix $H$ and the polynomial $p(x)$. As the degree of polynomial $p(x)$, $\ell$, gets higher, the matrix $p(H)$ may become denser. But at the same time, by the Cayley-Hamilton theorem, there exists a threshold $\ell_0\leq N$ such that the sparsity value is saturated for any polynomial of degree $\ell\geq\ell_0$, which depends upon the Hamiltonian matrix $H$. Therefore, in the worst-case scenario, for any polynomial $p(x)$ and a $\mathcal{O}(\text{poly}(n))$-sparse quantum Hamiltonian, the sparsity value $s$ of $p(H)$ scales as $\mathcal{O}((\text{poly}(n))^{\ell_0})$ for some fixed $\ell_0$. For the Hamiltonian models \eqref{eq:xxzmodel}, \eqref{eq:isingmodel}, and \eqref{eq:hmodel}, we numerically find that sparsity values are saturated at small degrees independent of $n$. In other words, by \cref{thm: convergence2}, we are able to obtain an almost linear scaling of the iteration complexity with the dimension $N$ for these models.

For example, in \cref{fig:Sparsity},
we observe quick saturations for the Hubbard model \eqref{eq:hmodel} and the XXZ model \eqref{eq:xxzmodel} within a degree of $10$ and a degree of $15$, respectively. That is, the sparsity values for these models may scale as $s=\mathcal{O}(n^{10})$ and $s=\mathcal{O}(n^{15})$, respectively. On the other hand, the usual representation of the TFIM \eqref{eq:isingmodel} is the case where we see no saturation with a sparsity value of $1$ reached within a degree of $10$ in all cases as in the bottom left panel of \cref{fig:Sparsity}. However, for this model, we can achieve saturation by simply applying a unitary transformation. For instance, we apply the Walsh Hadamard transform to the TFIM Hamiltonian \eqref{eq:isingmodel}, which yields
\begin{equation}\label{eq: Zisingmodel}
    H = J\sum_{j=1}^{n-1}X_jX_{j+1}+D\sum_{j=1}^nZ_j,
\end{equation}where $X$ and $Z$ components are interchanged, and see saturations within a degree of $10$ as in the bottom right panel of \cref{fig:Sparsity}. This may be due to the fact that the Pauli strings constituting only $Z$ components do not affect the sparsity in the computation of $p(H)$ and that the different representations \eqref{eq:isingmodel} and \eqref{eq: Zisingmodel} yield matrices with different sparsity patterns.

We remark that in general, the presence of threshold $\ell_0$ does not help reduce the per-iteration cost of estimating matrix elements quantumly or classically as discussed earlier in \cref{sec: matrix-element estimation}. Otherwise, we should know the reduced form of a matrix polynomial from the Cayley-Hamilton theorem, but it would demand an exponential computational cost with respect to $n$ in general.

\begin{figure}[htbp]

    \centering
    \begin{subfigure}[b]{0.4\textwidth}
    \includegraphics[width=\textwidth]{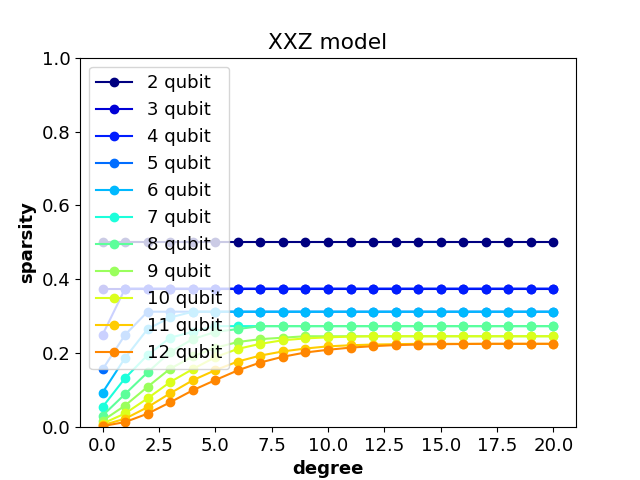}
    \end{subfigure}
    \begin{subfigure}[b]{0.4\textwidth}
    \includegraphics[width=\textwidth]{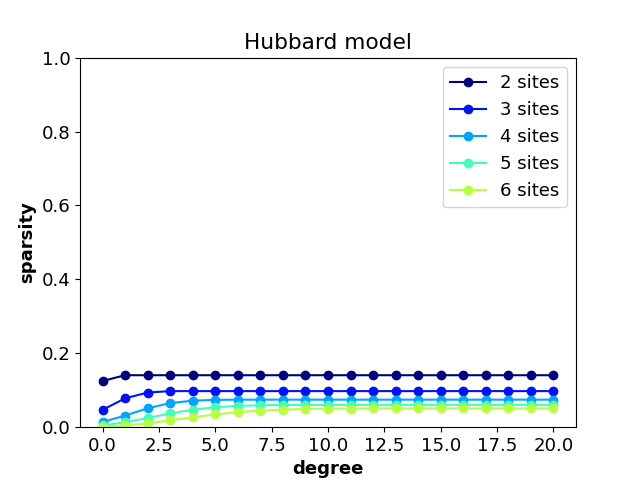}
    \end{subfigure}
    
    \begin{subfigure}[b]{0.4\textwidth}
    \includegraphics[width=\textwidth]{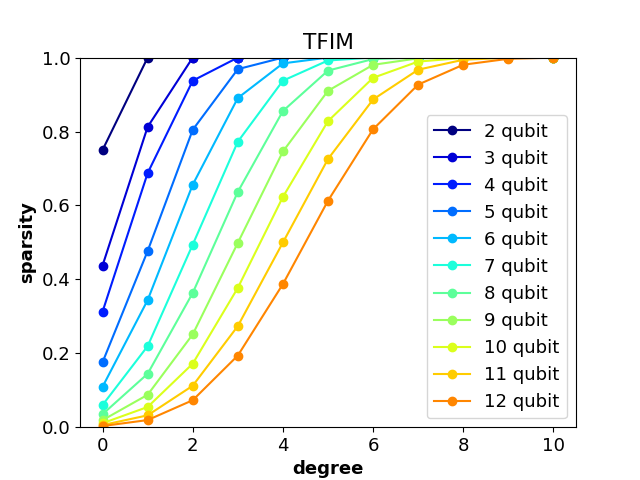}
    \end{subfigure}
    \begin{subfigure}[b]{0.4\textwidth}
    \includegraphics[width=\textwidth]{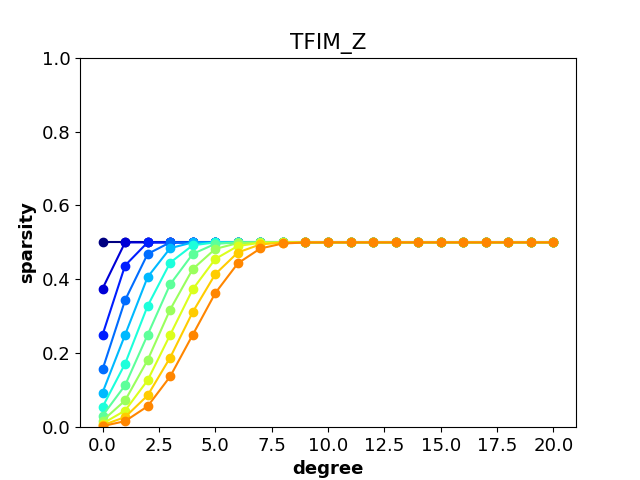}
    \end{subfigure}

    \caption{Plots of the sparsity with respect to the degree of the polynomial $p(x)$ for the model Hamiltonians \eqref{eq:xxzmodel}, \eqref{eq:hmodel}, and \eqref{eq:isingmodel}. The sparsity is defined to be the ratio $\frac{s}{N}$ where $p(H)$ is $s$-sparse for $H\in\mathbb{C}^{N\times N}$. The value $s$ is measured in such a way that the matrix $H$ is changed to one whose non-zero elements are set to $1$, the recursive relation that $H_{n+1}=HH_n+H_{n-1}$ with $H_2=H, H_1=I$ are used, which yields the same sparsity as $p(H)$ for a Chebyshev polynomial $p(x)$, and then $s$ is calculated as the maximum of the numbers of non-zero elements along the row of the resulting $H_\ell$ for  degree $\ell$. }
    \label{fig:Sparsity}
\end{figure}

\section{Conclusion}
In this work, we propose a quantum-classical hybrid random power method for ground-state computation. Unlike the previous inexact power methods \cite{chen2024quantum,lu2020full,shamir2016convergence, xu2018accelerated}, our method updates the iteration by estimating $p(H)\bm x$ in the expectation, where $p(x)$ is a filtering polynomial. In addition, the query complexity in our algorithm has a logarithmic dependence on the initial overlap, as it is based on the power method.  Here we used the Chebyshev polynomial  \cite{zhou2006self} for the filtering polynomial $p(x)$. The estimation of $p(H)\bm x$ is based on a randomized technique that requires only $\mathcal{O}(1)$ classical time and a polylogarithmic dependence of quantum circuit depth on the system size. To perform the estimation, our method implements a Hadamard test, given an access to either Hamiltonian simulation or block encoding, which is equivalent to sampling the elements of the matrix $p(H)$ with repetition. We prove that our algorithm outputs an approximation of ground state to any precision with probability one, as stated in \cref{thm: main result}. Especially when $p(H)$ is reasonably sparse matrix, which holds for the XXZ model, the Hubbard model, and the TFIM, our algorithm with the QSVT or the QETU converges with an almost linear scaling of the number of iterations with a quadratically improved quantum complexity in terms of query depth. Additionally, we prove the existence of a lower bound of the fidelity of the ground state scaling linearly with the magnitude of noise regardless of its characteristics, if it is less than half of the enlarged spectral gap of $p(H)$.

We highlight that the use of quantum polynomial filtering for estimating matrix elements provides us with a computational advantage that may not be obtained classically. Specifically, estimating the elements of $p(H)$ with only classical computations requires a cost exponentially scaling with the degree of the polynomial, while
doing so with the quantum filtering technique results in only a polynomial scaling cost.

There are a few questions to be further investigated. For
instance, can we prove a tighter lower bound than one in \cref{lem: overlap1} for the high fidelity against the perturbation error? Could we reduce the dependence of the spectral gap in the case of LCU approach \eqref{eq: esimatemat} as the quadratic improvement in the case of QSVT and QETU?  Another question is whether we can incorporate the phase estimation algorithms \cite{ni2023low,ding2023even} into our method to achieve a Heisenberg-scaling limit in samping matrix elements.

The code and data for our numerical results are available upon request.

\section{Acknowledgment}

We would like to thank Dr. Mancheon Han for helpful advice on numerical simulations. We thank Prof. Lin Lin for helpful comment. This research was supported by Quantum Simulator Development Project for Materials Innovation through the National Research Foundation of Korea (NRF) funded by the Korean government (Ministry of Science and ICT(MSIT))(No. NRF- 2023M3K5A1094813). We used resources of the Center for Advanced Computation at Korea Institute for Advanced Study and the National Energy Research Scientific Computing Center (NERSC), a U.S. Department of Energy Office of Science User Facility operated under Contract No. DE-AC02-05CH11231. S.C. was supported by a KIAS Individual Grant (CG090601) at Korea Institute for Advanced Study. T.K. is supported by a KIAS Individual Grant (No. CG096001) at Korea Institute for Advanced Study.

\crefalias{section}{appsec}
\section{Appendices}
\label[appsec]{sec: Appendix}

\subsection*{The proof of \cref{lem: overlap1}}
\label[appsec]{sec: proof of overlap1}
We use an outcome of the Weyl's inequality, which states that for Hermitians $A,B\in\mathbb{C}^{N\times N}$ and any $k\in\{1,2,..,N\}$, 
\begin{equation}\label{eq: outcomeweyl}
    \abs{\lambda_k(A+B)-\lambda_k(B)}\leq \norm{B},
\end{equation}where the $\lambda_k(A)$ denotes the $k$-th smallest eigenvalue of $A$ and similarly for $A+B$ and $B$.

To prove \cref{lem: overlap1}, we claim that 
    \begin{equation}\label{eq: indicator}
        P_{\lambda_1(\widetilde{A})}=\mathbb{I}_{[\lambda_1(A)-e,\lambda_1(A)+e]}\left(\widetilde{A}\right),
    \end{equation}where $\mathbb{I}$ is the indicator function that values one on the given interval and zero otherwise. To show this, we check that $\lambda_1(\widetilde{A})\in [\lambda_1(A)-e,\lambda_1(A)+e]$, and $\lambda_2(\widetilde{A})\not\in [\lambda_1(A)-e,\lambda_1(A)+e]$. By \eqref{eq: outcomeweyl}, $\lambda_1(\widetilde{A})\in [\lambda_1(A)-e,\lambda_1(A)+e]$. Again, by \eqref{eq: outcomeweyl}, $\lambda_2(\widetilde{A})\geq\lambda_2(A)-e$. By the assumption, $\lambda_2(A)-e>\lambda_1(A)+e$. Thus, the claim is proven. 
    
    Notice that the interval $[\lambda_1(A)-e,\lambda_N(A)+e]$ includes the spectra of $A$ and $\widetilde{A}$ by the Weyl's inequality, where $\lambda_N(A)$ is the largest eigenvalue of $A$. Restricting to this interval, by the Stone–Weierstrass theorem, we can take a sequence of polynomials $\{p_n(x)\}_{n=1}^{\infty}$ that uniformly converges to the indicator function in \eqref{eq: indicator}. Then for any $\epsilon$, there exists a $n>0$ such that $p_n$ is $\frac{\epsilon}{2}$-close to the indicator function on the interval in the uniform norm. Thus, it immediately follows that
    \begin{equation}
        \norm{\mathbb{I}(\widetilde{A})-p_n(\widetilde{A})}\leq\frac{\epsilon}{2},\quad \norm{\mathbb{I}(A)-p_n(A)}\leq\frac{\epsilon}{2}, 
    \end{equation}where $\mathbb{I}:=\mathbb{I}_{[\lambda_1(A)-e,\lambda_1(A)+e]}$.
    Together with the above claim, we observe
    \begin{equation}\label{eq: eqs}
        \begin{split}
            P_{\lambda_1(\widetilde{A})}&=\mathbb{I}(\widetilde{A})\\
            &=p_n(\widetilde{A})+\mathbb{I}(\widetilde{A})-p_n(\widetilde{A})\\
            &=p_n(A)+R_n(A,E)+\mathbb{I}(\widetilde{A})-p_n(\widetilde{A})\\
            &=\mathbb{I}(A)+R_n(A,E)+p_n(A)-\mathbb{I}(A)+\mathbb{I}(\widetilde{A})-p_n(\widetilde{A})\\
        &=P_{\lambda_1(A)}+R_n(A,E)+p_n(A)-\mathbb{I}(A)+\mathbb{I}(\widetilde{A})-p_n(\widetilde{A})
        \end{split}
    \end{equation}where $R_n(A,E):=p_n(\widetilde{A})-p_n(A)$. We note that $\mathbb{I}(A)=P_{\lambda_1(A)}$ by the selection of $e$. In addition, $\norm{R_n(A,E)}\leq C\norm{E}$, since all terms in $R_n(A,E)$ have at least one $E$ as a factor, and $C$ is determined by $\norm{A}$, $e$, and $\epsilon$, based on the triangle inequality argument. Let $\ket{\widetilde{\bs q}}=\sum_{j=1}^g\beta_j\ket{\widetilde{\bs q}_j}$ be a ground state of $\widetilde{A}$. By applying the braket of the state $\ket{\widetilde{\bs q}}$ to \eqref{eq: eqs}, we therefore obtain
    \begin{equation}
        1\leq \bra{\widetilde{\bm q}}P_{\lambda_1(A)}\ket{\widetilde{\bm q}}+C\norm{E} +\epsilon,
    \end{equation}which completes the proof.

\subsection*{An exponential growth of the spectral gap with the degree of Chebyshev polynomial(section order change)}\label[appsec]{sec: proof of enlarged spectral gap}

In this section, we show that the Chebyshev polynomial filtering, $p_\ell(H)$, enlarges the spectral gap of $\overline{H}$ exponentially with respect to degree $\ell$. Recall that $\overline{H}$ denotes the transformed Hamiltonian in \eqref{Htilde}. In the following, we denote the eigenvalue of $\overline{H}$ by $\overline{\lambda}$ according to \eqref{Htilde}.

\begin{proposition}\label{prop: spectral gap}
Assume that $\lambda_{lb}\in(\lambda_1,\lambda_2]$ and $\lambda_{ub}>\lambda_N$. Then,
\begin{equation}
    \abs{p_\ell(\overline{\lambda}_1)}=\Theta\left(\left(\abs{\overline{\lambda}_1}+\sqrt{\abs{\overline{\lambda}_1}^2-1}\right)^\ell\right).
\end{equation}
Especially, 
\begin{equation}
    \abs{p_\ell(\overline{\lambda}_1)-p_\ell(\overline{\lambda}_2)}=\Theta\left(\left(\abs{\overline{\lambda}_1}+\sqrt{\abs{\overline{\lambda}_1}^2-1}\right)^\ell\right),
\end{equation}where $\overline{\lambda}=\frac{2}{\lambda_{ub}-\lambda_{lb}}(\lambda-\lambda_{ub})+1$ and $\overline{\lambda}_1>1$.
With the relaxed condition that $\lambda_{lb}\in(\lambda_2,\lambda_N]$, the following hold 
\begin{eqnarray}
    &\abs{p_\ell(\overline{\lambda}_1)}=\Theta\left(\left(\abs{\overline{\lambda}_1}+\sqrt{\abs{\overline{\lambda}_1}^2-1}\right)^\ell\right)\\
    &\abs{p_\ell(\overline{\lambda}_2)}=\Theta\left(\left(\abs{\overline{\lambda}_2}+\sqrt{\abs{\overline{\lambda}_2}^2-1}\right)^\ell\right),
\end{eqnarray}and
\begin{equation}
    \abs{p_\ell(\overline{\lambda}_1)-p_\ell(\overline{\lambda}_2)}=\Theta\left(\left(\abs{\overline{\lambda}_1}+\sqrt{\abs{\overline{\lambda}_1}^2-1}\right)^\ell\right).
\end{equation}

\end{proposition}

\begin{proof}
    We use the identity for the Chebyshev polynomial $p_{\ell}(x)$,
\begin{equation}\label{identity}
    p_{\ell}(x) = \frac{(x-\sqrt{x^2-1})^\ell+(x+\sqrt{x^2-1})^\ell}{2}.
\end{equation}We consider the first case that $\lambda_{lb}$ is given in $(\lambda_1,\lambda_2]$. Notice that in this case, $\overline{\lambda}_2\in(-1,1)$ and $\overline{\lambda}_1<-1$ where $\overline{\lambda}:=\frac{2}{\lambda_{ub}-\lambda_{lb}}(\lambda-\lambda_{ub})+1$. Then we have 
\begin{equation}
    \abs{p_\ell(\overline{\lambda}_2)}\leq 1,
\end{equation}and
\begin{equation}
    \abs{p_\ell(\overline{\lambda}_1)}=\abs{p_\ell(\abs{\overline{\lambda}_1})}\geq \frac{1}{2}\abs{\abs{\overline{\lambda}_1}+\sqrt{\abs{\overline{\lambda}_1}^2-1}}^\ell.
\end{equation}
Thus, the triangle inequality implies
\begin{equation}
\begin{split}
    \abs{p_\ell(\overline{\lambda}_1)-p_\ell(\overline{\lambda}_2)}&\geq \frac{1}{2}\abs{\abs{\overline{\lambda}_1}+\sqrt{\abs{\overline{\lambda}_1}^2-1}}^\ell-1.
\end{split}
\end{equation}Similarly, one can check the second case that $\lambda_{lb}$ is given in $(\lambda_2,\lambda_N]$, and thus $\overline{\lambda}_1<\overline{\lambda}_2<-1$, by using the identity \eqref{identity}.
\end{proof}

\subsection*{Approximate Chebyshev filtering via a Fourier-series approximation (title changed)}\label[appsec]{sec: fejer}  

To implement the Chebyshev filtering \cite{zhou2006self} using Hamiltonian simulation in \cref{alg: algorithm1}, we first transform the Hamiltonian as \eqref{Htilde}, and normalize the transformed matrix such that its spectrum belongs to $[-\pi,\pi]$. Denote the transformed matrix by $\overline{H}$. Let $\alpha$ satisfy $\norm{\overline{H}/\alpha}\leq \pi$. Given access to Hamiltonian simulation $e^{iHt}$, the LCU representation \eqref{eq: esimatemat} enables us to estimate elements of the matrix polynomial $p_\ell(\overline{H})$ approximately. In practical calculations, the Fe\'{j}er kernel often shows more robust and stable performance than the Dirichlet kernel \cite{weisse2006kernel}. Our numerical experience also showed that the Fe\'{j}er kernel was better than the Dirichlet kernel, requiring smaller values of $M$ in \eqref{eq: esimatemat} in all cases in \cref{fig:TFIMXXZ} and thereby yielding smaller numbers of circuit sampling.

In the following, we provide a formula of Fourier coefficients when the Fe\'{j}er kernel is considered to approximate $p_\ell(\alpha x)$. The convolution form of Fourier-series approximation using Fe\'{j}er kernel is denoted as
\begin{equation}\label{eq: fourierway}
    p_\ell(\alpha x)\approx \sum_{m=-M}^Mc_me^{imx} = (p_\ell(\alpha (\cdot))*S_M)(x),
\end{equation}where $c_m=\frac{M-m+1}{2\pi M}\int_{-\pi}^\pi p_\ell(\alpha x)e^{-imx}dx$, and $S_M(x) = \frac{\sin^2\left(Mx/2\right)}{M\sin^2(x/2)}$, the Fe\'{j}er kernel. By the Fe\'{j}er's theorem, the RHS converges uniformly to the LHS as $M\rightarrow \infty$. The $c_m$ can be computed via a recursive relation for the following auxiliary variable,
\begin{equation}
    c(m,n) = \frac{1}{2\pi}\int_{-\pi}^\pi  x^n e^{-imx}dx.
\end{equation}By the integral of parts, we have
\begin{equation}
\begin{split}
    c(m,n+1) &= \frac{i}{2\pi m}\left[\pi^{n+1}e^{-im\pi}-(-\pi)^{n+1}e^{im\pi}\right]+\frac{(n+1)}{im}c(m,n)\\
    c(0,n) &=\begin{cases}
        0\\
        \frac{\pi^n}{n+1}
    \end{cases}.
\end{split}
\end{equation}Then, we compute $c_m$ as follows,
\begin{equation}
    c_m = \frac{M-m+1}{M}\sum_{j=0}^{\ell}\alpha^j d_jc(m,j),
\end{equation}where $d_j$ is the $j$-th coefficient of the Chebyshev polynomial $p_\ell(x)$. 
Once the $c_m$ is computed, we can estimate the matrix elements of $p_\ell(\overline{H})$ by replacing $x$ with $\overline{H}/\alpha$ in \eqref{eq: fourierway} and performing the Hadamard test in \cref{fig:circuit}.

When the Fe\'{j}er kernel is used in the LCU approach \eqref{eq: esimatemat}, the following lemma shows that $M=\mathcal{O}(\Delta^{-2})$ where  $\Delta=\lambda_2-\lambda_1$,  together with the existence of a Chebyshev filtering polynomial of degree $\mathcal{O}(\Delta^{-\frac{1}{2}})$.

\begin{lemma}\label{lem: fejerkernel}
    Assume the knowledge of filtering bounds in \eqref{Htilde} such that $\lambda_{ub}\in[\lambda_N, \lambda_N+\epsilon(\lambda_N-\lambda_2)]$ for some $\epsilon\in(0,\frac{1}{\Delta})$ and $\lambda_{lb}=\frac{c_\lambda \lambda_{ub}+\lambda_2}{c_\lambda+1}$ for some $c_\lambda\in(0,\frac{1}{\epsilon})$. If $\frac{\Delta}{\lambda_{ub}-\lambda_2}\ll 1$, for any $c_\lambda\leq \Delta$, there exist a Chebyshev filtering polynomial for \cref{alg: algorithm1}, $p_\ell(x)$, and a Fourier-series approximation with Fe\'{j}er kernel, $(p_\ell*S_M)(x)$, such that 
    \begin{itemize}
        \item[(i)] $\ell=\mathcal{O}(\Delta^{-\frac{1}{2}}), \quad\abs{p_\ell(\overline{\lambda}_1)-p_\ell(\overline{\lambda}_2)}=\Omega(1),\quad \abs{p_\ell(\overline{\lambda}_1)}=\mathcal{O}(1)$,
        \item[(ii)] $M=\mathcal{O}(\Delta^{-2}), \quad\abs{(p_\ell*S_M)(\overline{\lambda}_1)-(p_\ell*S_M)(\overline{\lambda}_2)}=\Omega(1),\quad \abs{(p_\ell*S_M)(\overline{\lambda}_1)}=\mathcal{O}(1)$.
    \end{itemize}
\end{lemma}

\begin{proof}
    
Without loss of generality, we assume that $\ell$ is even, which ensures that $p_\ell(\overline{\lambda}_1)>p_\ell(\overline{\lambda}_2)>1$ by the assumption and the linear transformation in \eqref{Htilde}. Using \eqref{identity}, we can estimte a lower bound of the transformed spectral gap as
\begin{equation}
    \begin{split}
    \abs{p_\ell(\overline{\lambda}_1)-p_\ell(\overline{\lambda}_2)}
    &\geq \left(\overline{\lambda}_2-\overline{\lambda}_1\right) \abs{p_\ell'(\overline{\lambda}_2)}\\
    &\geq \left(\overline{\lambda}_2-\overline{\lambda}_1\right)\ell^2\abs{\overline{\lambda}_2}^{\ell-1}\\
    &\geq \frac{\Delta}{\lambda_{ub}-\lambda_2}\ell^2.
\end{split} 
\end{equation}The first inequality holds by the mean value theorem and the monotonicity of the Chebyshev polynomial outside $[-1,1]$ and $\overline{\lambda}_1<\overline{\lambda}_2<-1$.  The second inequality can be verified by differentiating the identity \eqref{identity} and taking the term involving $x^{\ell-1}$. The last line holds by the definition of linear transformation \eqref{Htilde}. The resulting inequality means that $\Omega(1)$ spectral gap is obtained when $\ell =\mathcal{O}(\Delta^{-\frac{1}{2}})$. Moreover, using \eqref{identity} again,  we can deduce that for small $\epsilon>0$,
\begin{equation}
\begin{split}
    \abs{p_\ell(1+\epsilon)}\leq (1+\epsilon+\sqrt{(1+\epsilon)^2-1})^\ell\leq (1+3\sqrt{\epsilon})^\ell,
\end{split}
\end{equation}which implies that ,
\begin{equation}\label{eq: ub of plambda1}
\begin{split}
    \abs{p_\ell(\overline{\lambda}_1)}&\leq  \left(1+3\sqrt{c_\lambda+\frac{\Delta}{\lambda_{ub}-\lambda_{lb}}}\right)^\ell\leq \left(1+3\sqrt{\left(1+\frac{1}{\lambda_{ub}-\lambda_{lb}}\right)\Delta}\right)^\ell\leq \mathcal{O}(1) 
\end{split}
\end{equation}by noticing that $\overline{\lambda}_1=-1-\left(c_\lambda+\frac{\Delta}{\lambda_{ub}-\lambda_{lb}}\right)$, $c_\lambda\leq\Delta$, and $\ell=\mathcal{O}(\Delta^{-\frac{1}{2}})$. The last inequality can be easily checked by the numeric inequality $(1+x)^{\frac{1}{x}}\leq e$ for any $x>0$. We have proven the first statement.

To prove the second statement, we use a more precise asymptotic behavior of $p_\ell'(x)$ than the above inequality when $\Delta\ll 1$, that is, if $x=1+\Delta$ and $\ell=C\Delta^{-\frac{1}{2}}$ with $C>0$,
\begin{equation}\label{eq: p'ineq}
    p_\ell'(x) = C\frac{(1+\Delta+\sqrt{2\Delta+\Delta^2})^{\frac{C}{\sqrt{\Delta}}}-(1+\Delta-\sqrt{2\Delta+\Delta^2})^{\frac{C}{\sqrt{\Delta}}}}{\Delta\sqrt{\Delta+2}}\approx \frac{C(e^{C\sqrt{2}}-e^{-C\sqrt{2}})}{\sqrt{2}\Delta}.
\end{equation}To proceed, we construct a smooth periodic extension of the polynomial $p_\ell(x)$ to be defined on $[-\pi,\pi]$, which is the requirement for using the properties of Fe\'{j}er kernel in analysis. The extension of $p_\ell(x)$ can be constructed as
\begin{subequations}
    \begin{eqnarray}
        &p_\ell(-\overline{\lambda}_1+x):= 2p_\ell(-\overline{\lambda}_1)- p_\ell(-\overline{\lambda}_1-x) \\
        &p_\ell(y)=p_\ell(-2\overline{\lambda}_1)\\
        &p_\ell(x) = p_\ell(-x)  
    \end{eqnarray}
\end{subequations}for any $x\in[0,-\overline{\lambda}_1]$ and any $y\in[-2\overline{\lambda}_1,\pi]$. That is, the function $\overline{p}_\ell(x)$ is symmetrically defined from $p_\ell(x)$ about the point $(-\overline{\lambda}_1,p_\ell(\overline{-\lambda}_1))$. Now recalling that $\overline{\lambda}_1=-1-\mathcal{O}(\Delta)$ by the assumption, we obtain that for any $x\in[\overline{\lambda}_1,-1]$
\begin{equation}\label{eq: ineq error bound}
\begin{split}
    \abs{p_\ell*S_M(x)-p_\ell(x)}&\leq \frac{1}{2\pi}\int_{-\pi}^\pi \abs{S_M(y)}\abs{p_\ell(x-y)-p_\ell(x)}dy\\
    &=\int_{\abs{y}<\delta}\abs{S_M(y)}\abs{p_\ell(x-y)-p_\ell(x)}dy+\int_{\delta\leq \abs{y}\leq \pi}\abs{S_M(y)}\abs{p_\ell(x-y)-p_\ell(x)}dy\\
    &\leq \delta \sup_{\abs{y}<\delta}\abs{p_\ell'(x-y)} +\frac{2\sup_{\abs{x}\leq\pi} \abs{p_\ell(x)}}{M\delta^2}\\
    &\leq \delta \frac{C'}{\Delta} +\frac{4\sup_{\abs{x}\leq-\overline{\lambda}_1} \abs{p_\ell(x)}}{M\delta^2}.
\end{split}
\end{equation}$C'$ is an absolute constant, which can be deduced from \eqref{eq: p'ineq}. For the second inequality, we used the decaying property of Fe\'{j}er kernel with respect to $M$.  As a result, to some tolerance $\epsilon_{tol}$, it follows that for $\delta=\frac{\Delta}{C\epsilon_{tol}}$ and $M=\frac{4\sup_{\abs{x}\leq-\overline{\lambda}_1}\abs{p_\ell(x)}}{\delta^2\epsilon_{tol}}$, $\abs{p_\ell*S_M(x)-p_\ell(x)}\leq 2\epsilon_{tol}$. This means that when $M=\mathcal{O}(\Delta^{-2})$, the Fourier-series approximation with Fe\'{j}er kernel is reasonable. In addition, using the property of Chebyshev polynomial that $\abs{p_\ell'(x)}\leq \ell^2$
for any $x\in[-1,1]$, we can also derive a similar error bound in $x\in[-1,1]$ for the Fourier-series approximation by setting $\ell=\mathcal{O}(\Delta^{-\frac{1}{2}})$ and $M=\mathcal{O}(\Delta^{-2})$, as we did in \eqref{eq: ineq error bound}. Therefore, the second statement is proven by setting  $M=\mathcal{O}(\Delta^{-2})$ for some small tolerance $\epsilon_{tol}$ such that the small tolerance is relatively negligble compared to the quantatitve properties in the first statement.

\end{proof}

\subsection*{Complexity of evaluating matrix elements using QSP techniques (title changed)}\label[appsec]{sec: QSP}

First, We discuss how to estimate elements of matrix $p_\ell(H)$ using QSVT \cite{GSLW19} with a block encoding of $H$ . The result can be straightforwardly applied to the case of QETU with Hamiltonian simulation \cite{dong2021efficient}, as we will point out.

The definition of a block-encoding of $H$ is as follows: for a matrix $H$ with $\norm{\frac{H}{\alpha}}\leq 1$, we say that $U_H$ is an $(\alpha,a,\epsilon)$-block-encoding of $H$ if $U_H$ is a $(m + a)$-qubit unitary, and
\begin{align}\label{be-H}
  \norm{H - \alpha  (\bra*{0^{\otimes a}}\otimes I)U_H(\ket*{0^{\otimes a}}\otimes I)} \leq \epsilon.
\end{align}
For a given block-encoding of a Hermitian matrix, the quantum singular value transformation (QSVT) \cite{GSLW19}  prepares a corresponding matrix polynomial to be encoded into the left corner of a unitary, namely, 
\begin{equation}\label{be-p}
    U_{p_\ell(H)/C} = \begin{pmatrix}
        p_\ell(H)/C & \vdot \\
        \vdot & \vdot\\
    \end{pmatrix},
\end{equation}where $C$ normalizes the matrix polynomial $p_\ell(H)$. Together with \cite[Theorem 31]{GSLW19}, we obtain the following overall complexity for the estimation of matrix element,
\begin{theorem}[Overall complexity]\label{thm:complexity}
For any tolerance $\epsilon>0$ and failure probability $\delta>0$,  if the sample size is 
$S=\frac{4C^2}{\epsilon^2}\log\frac{4}{\delta}$,
then the Hadamard test in \cref{fig:circuit} outputs an $\epsilon$-close unbiased estimate $\xi$ of the matrix element $\bra{\bs i}p_\ell(H)\ket{\bs j}$ with confidence $1-\delta$, namely,
\begin{eqnarray}
&\mathbb{E}[\xi]=\bra{\bs i}p_\ell(H)\ket{\bs j},\\
&\mathbb{P}\left(\abs{\xi-\bra{\bs i}p_\ell(H)\ket{\bs j}}\geq\epsilon\right)\leq\delta.
\end{eqnarray}Each sample is obtained from $\ell$ applications of $U_H$ and $U_H^{\dagger}$, one application of controlled-$U_H$ gate, and $\mathcal{O}(\ell)$ applications of other one- and two-qubit gates.
    
\end{theorem}

\begin{proof}
Define a random variable
\begin{equation}\label{unbiasedxi}
    \xi = \frac{1}{S}\left(\sum_{j=1}^Sz_j+iz_{j+S}\right),
\end{equation}where the $z_j$'s and $z_{j+S}$ are independent random variables in $\{-1,1\}$ obtained from $2S$ applications of the Hadamard test \cref{fig:circuit} by letting $W=I$ and $S^{\dagger}$, respectively. The values $1$ and $-1$ correspond to the cases that the state in the first qubit results in $0$ and $1$, respectively.  
Notice that when $W=I$, the probability of state in the ancilla qubit in \cref{fig:circuit} satisfies
\begin{equation}\label{eq: mateleest}
    \begin{split}
        1-2\mathbb{P}\left(\text{0}\right)&=\mathrm{Re}\left(
        \bra{0^{n+m+1}}(I\otimes X_i^{\dagger})U_{p(H)}(I\otimes X_j)\ket{0^{n+m+1}}\right)\\
        &=\mathrm{Re}\left(\bra{0^{n+m+1}}(I\otimes X_i^{\dagger})\ket{0}\left(\ket{0^m}p(H)\ket{\bs j}+\ket{\perp}\right)\right)\\
        & = \mathrm{Re}\left(\bra{0^{m}}\bra{\bs i}\left(\ket{0^m}p(H)\ket{\bs j}+\ket{\perp}\right)\right)\\
        & = \mathrm{Re}\left(\bra{\bs i}p(H)\ket{\bs j}\right),
    \end{split}
\end{equation}where $\ket{\perp}$ is perpendicular to any state of the form $\ket{0^m}\ket{\bs x}$. Similarly, the imaginary part of $\bra{\bs i}p(H)\ket{\bs j}$ can be estimated by setting $W=S^\dagger$ in \cref{fig:circuit}. We observe that $\xi$ is a sample mean for the matrix element, and 
therefore we arrive at
\begin{equation}
    \mathbb{E}[\xi]=\bra{\bs i}p_\ell(H)\ket{\bs j}.
\end{equation}By the Hoeffding's inequality, we have
\begin{equation}
    \mathbb{P}\left(\abs{\mathrm{Re}(\xi)-\mathrm{Re}\left(\bra{\bs i}p_\ell(H)\ket{\bs j}\right)}\geq \epsilon\right)\leq 2\exp\left(-\frac{S\epsilon^2}{2}\right).
\end{equation}The same inequality applies to the case of the imaginary part of $\bra{\bs i}p_\ell(H)\ket{\bs j}$. Since the inequality that $\abs{a+bi-x-yi}\geq\epsilon$ for real $a,b,x,y$ implies that $\abs{a-x}\geq\frac{\epsilon}{\sqrt{2}}$ or $\abs{b-y}\geq\frac{\epsilon}{\sqrt{2}}$, we have 
$\mathbb{P}\left(\abs{\xi-\bra{\bs i}p_\ell(H)\ket{\bs j}}\geq \epsilon\right)\leq 4\exp\left(-\frac{S\epsilon^2}{4}\right)$.
In the case that $p(H)=\frac{p_\ell(\overline{H})}{C}$, it immediately follows
\begin{equation}
    \mathbb{P}\left(\abs{\widetilde{\xi}-\bra{\bs i}p_\ell(H)\ket{\bs j}}\geq \epsilon\right)\leq 4\exp\left(-\frac{S\epsilon^2}{4C^2}\right),
\end{equation}where $\widetilde{\xi}$ is defined similar to \eqref{unbiasedxi} with random variables chosen from $\{-C,C\}$ and $C$ defined in \eqref{be-p}. 
\end{proof}
Simiarly, we can do the task with QETU circuit \cite{dong2021efficient}. We note that
\cite[Theorem 1]{dong2021efficient} can be directly applied to \cref{thm:complexity} by considering $p_\ell(x)=(p_\ell\circ g)(\cos \frac{x}{2})$ where $g(x)=2\arccos(x)$. Then together with \eqref{eq: mateleest}, we can estimate elements of matrix $p_\ell(H)$ with a circuit in \cite{dong2021efficient} via the Hadamard test \cref{fig:circuit} with two ancilla qubits where one is used due to the Hadamard test and the other is defined in the QETU circuit \cite{dong2021efficient}.

\subsection{The proof of \cref{proposition}}\label[appsec]{sec: proof of proposition}

We show the first property as
\begin{equation}
    \mathbb{E}[\bs g] = \mathbb{E}_c\left[\mathbb{E}_r[\mathbb{E}_\xi[\bs g]]\right] = \mathbb{E}_c[\sum_{i_c}\frac{m_r}{N}(\bs x,H\bs e_{i_c})\bs e_{i_c}]=\frac{m_rm_c}{N^2}H\bs x.
\end{equation}The first equality holds since there are three types of randomness: sampling non-zero components of $\bs x$, computational basis element $\bs e_{i_c}$, and estimating the corresponding matrix element. In the second equality, we notice that the noise from quantum computation is averaged out and that 
\begin{equation}
    \mathbb{E}[\sum_{i_r}x_{i_r}\bs e_{i_r}]=\frac{{N-1 \choose m_r-1}}{{N \choose m_r}}\bs x = \frac{m_r}{N}\bs x.
\end{equation}Similarly, the third equality holds by considering the factor $\frac{{N-1\choose m_c-1}}{{N\choose m_c}}$. 

Next we prove the second property. 
We first consider the case $m_r=N$. In this case, we can rewrite the estimator as
\begin{equation}
    \bs g = \sum_{i_c}(\bs x,H\bs e_{i_c})\bs e_{i_c}.
\end{equation}We observe that
\begin{equation}
    \begin{split}
        \mathbb{E}[\bs g\bs g^*] & = \mathbb{E}\left[\left(\sum_{i_c}(\bs x,H\bs e_{i_c})\bs e_{i_c}\right)\left(\sum_{i_c}(\bs x,H\bs e_{i_c})\bs e_{i_c}  \right)^*\right]\\
        & = \mathbb{E}\left[\sum_{i_c,j_c}(\bs x,H\bs e_{i_c})(H\bs e_{j_c},\bs x)\bs e_{i_c}\bs e_{j_c}^*\right]\\
        & = \mathbb{E}\left[\sum_{i_c}\abs{(\bs x,H\bs e_{i_c})}^2\bs e_{i_c}\bs e_{i_c}^*+\sum_{i_c\neq j_c}(\bs x,H\bs e_{i_c})(H\bs e_{j_c},\bs x)\bs e_{i_c}\bs e_{j_c}^*\right] \\
        & = \frac{{N-1\choose m_c-1}}{{N\choose m_c}}\mathrm{diag}(H\bs x\bs x^*H)+\frac{{N-2\choose m_c-2}}{{N\choose m_c}}(H\bs x\bs x^*H-\mathrm{diag}(H\bs x\bs x^*H))\\
        & = \frac{m_c(N-m_c)}{N(N-1)}\mathrm{diag}\left(H\bs x \bs x^*H\right)+\frac{m_c(m_c-1)}{N(N-1)}H\bs x\bs x^*H. 
    \end{split}
\end{equation}By this result, we show the second property for the case that $m_r\leq N$, obtaining that  
\begin{equation}
\begin{split}
    \mathbb{E}[\bs g\bs g^*] &= \mathbb{E}_{c}\left[\mathbb{E}_{r}[\bs g \bs g^*]\right]\\
    &=\mathbb{E}_{c}\left[\mathbb{E}_{r}[(\bs g - \mathbb{E}_{r}[\bs g])(\bs g- \mathbb{E}_{r}[\bs g])^*]+\mathbb{E}_{r}[\bs g]\mathbb{E}_{r}[\bs g]^*\right]\\
    & = \mathbb{E}_c\left[\mathbb{E}_{r}[(\bs g - \mathbb{E}_{r}[\bs g])(\bs g- \mathbb{E}_{r}[\bs g])^*]\right] \\
    &+ \frac{m_c(N-m_c)}{N(N-1)}\mathrm{diag}\left(H\bs x \bs x^*H\right)+\frac{m_c(m_c-1)}{N(N-1)}H\bs x\bs x^*H,
\end{split}
\end{equation}where the conditional expectation w.r.t $r$, $\mathbb{E}_{r}[\bs g]$, depends on the sampled column index $c$. In the second and third equalities, we used the identities $\mathbb{E}[\bs v\bs v^*]=\mathbb{E}[(\bs v-\mathbb{E}[\bs v])(\bs v-\mathbb{E}[\bs v])^*]+\mathbb{E}[\bs v]\mathbb{E}[\bs v]^*$ and $\mathbb{E}_{r}[\bm g]=\sum_{i_c} (\bm x,H\bm e_{i_c})\bm e_{i_c}$. We define
\begin{equation}\label{eq: Sigmarc}
    \Sigma_{r,c}=\mathbb{E}_c[\mathbb{E}_{r}\left[(\bs g - \mathbb{E}_{r}[\bs g])(\bs g- \mathbb{E}_{r}[\bs g])^*\right]].
\end{equation}Denote $\bm x_r=\sum_{i_r}(\bm x,e_{i_r})e_{i_r}$ to be the vector which is sparsified from $\bm x$ and has the non-zero elements associated to the $m_r$ sampled indices $\{i_r\}$. We notice that
\begin{equation}\label{eq: Sigmarc prop}
    \begin{split}
        \mathrm{tr}(\Sigma_{r,c})&= \mathbb{E}_c\left[\mathbb{E}_r\left[\sum_{i_c}\abs{(\bs x_r-\bs x,H\bs e_{i_c})}^2\right]\right]\\
        &=\mathbb{E}_c\left[\sum_{i_c}\mathbb{E}_r\left[\abs{(\bs x_r-\bs x,H\bs e_{i_c})}^2\right]\right]\\
    &=\frac{m_c}{N}\sum_{i=1}^N\mathbb{E}_r\left[\abs{(\bs x_r-\bs x,H\bs e_{i})}^2\right]\\
    &= \frac{m_c}{N}\mathbb{E}_r\left[\norm{H(\bs x_r-\bs x)}^2\right]\\
    &\leq \frac{m_c}{N}\norm{H}^2\mathbb{E}_r\left[\norm{\bs x_r-\bs x}^2\right]\\
    &\leq \frac{m_c(N-m_r)}{N^2}\norm{H}^2\norm{\bm x}^2
    \end{split}
\end{equation}
The third equality holds since there are $\frac{{N-1\choose m_c-1}}{{N\choose m_c}}$ events when the index $i_c$ is sampled. The fourth equality holds since $\sum_{i=1}^Ne_ie_i^*=I$.  The last inequality holds by the definition of $\bm x_r$ whose $m_r$ non-zero elements are sampled randomly from $\bm x$.

\subsection{Convergence for the case of exact computation with $m_r=N, m_c=1$ in \eqref{unbiasedest}}\label[appsec]{sec: proof of convergence thm}
We consider the simple case that \eqref{unbiasedest} is estimated exactly with $m_r=N, m_c=1$. The analysis of this case will be extended to the case of noisy computation \eqref{unbiasedest} in \cref{sec: proof of convergence thm2}. Let us define an estimate as follows,
\begin{equation}\label{eq: semi-unbiased}
    \bs g = (H\bs x,\bs e_{i_s})\bs e_{i_s},
\end{equation}where $i_s$ denotes the index randomly selected from $\{1,..,N\}$. By definition, notice that
\begin{equation}
    \mathbb{E}[\bs g] = \frac{1}{N}H\bs x.
\end{equation}

\begin{theorem}\label{thm: convergence}
    Assume the estimator \eqref{eq: semi-unbiased}. For any precision $\epsilon>0$ and $T=\mathcal{O}\left(\frac{1}{\log c_{\epsilon,a}}\log\frac{1}{\abs{(\psi_{GS},\bs x_1)}}\right)$, the overlap of the last iterate from \cref{alg: algorithm1} with the ground state is greater than $1-\epsilon$ with probability one, namely,  
    \begin{equation}
     \mathbb{P}\left(\frac{\abs{(\psi_{GS},\bs x_{T})}^2}{\|\bs x_{T}\|^2}>1-\epsilon\right)=1.
    \end{equation}
\end{theorem}

\begin{proof}
In this proof and thereafter, we denote $\mathbb{E}_k$ to be expectation conditioned on the $k$-th iteration, $\bm x_k$. We observe that the Cauchy-Schwarz inequality yields
\begin{equation}\label{eq: recursive ineq}
\begin{split}
    \mathbb{E}_k\left[\frac{\abs{(\psi_{GS},\bs x_{k+1})}^2}{\|\bs x_{k+1}\|^2}\right]&\geq \frac{\abs{\mathbb{E}_k[(\psi_{GS},\bs x_{k+1})]}^2}{\mathbb{E}_k[\|\bs x_{k+1}\|^2]}\\
    & = \frac{\abs{\mathbb{E}_k[(\psi_{GS},\bs x_{k+1})]}^2}{\sum_{j=1}^N\mathbb{E}_k[\abs{(\psi_{j},\bs x_{k+1})}^2]},
\end{split}
\end{equation}assuming deterministically non-zero denominators.  For sufficiently small step size $a>0$, the norm of $\bs x_k$ is bounded below by some positive value deterministically. Using \eqref{eq: ub of g}, we can easily check this since
\begin{equation}\label{eq: gradbound}
\begin{split}
    &\norm{\bs g_k}\leq \norm{H}\norm{\bs x_k}\\
    &\Longrightarrow\norm{\bs x_k}\geq (1-\norm{H}a)^{k-1}>0,
\end{split}
\end{equation}if $a<\frac{1}{\norm{H}}$. In the last line, we used the recursive relation \eqref{sgd}, the initial condition $\norm{\bm x_1
}=1$, and the triangle inequality. The following analysis is performed assuming $a<\frac{1}{\norm{H}}$.

First, we estimate the nominator and the denominator of the right side in \eqref{eq: recursive ineq}. From \eqref{sgd} and \eqref{eq: semi-unbiased}, notice that
\begin{equation}\label{eq: nominator}
\begin{split}
    &\mathbb{E}_k[(\psi_{GS},\bs x_{k+1})]\\
    &=(\psi_{GS},\mathbb{E}_k[\bs x_{k+1}])\\
    &=(\psi_{GS}, \bs x_k) - a(\psi_{GS},\mathbb{E}_k[\bs g_k])\\
    &=(\psi_{GS},\bs x_k) - a(\psi_{GS},\frac{1}{N}H\bs x_k)\\
    &=\left(1-\frac{\lambda_1}{N}a\right)(\psi_{GS},\bs x_k).
\end{split} 
\end{equation}
In addition, we have 
\begin{equation}
    \mathbb{E}_k[\abs{(\psi_{j},\bs x_{k+1})}^2]=\abs{(\psi_{j},\bs x_k)}^2-2a\mathbf{Re}[(\psi_{j},\bs x_k)\mathbb{E}_k[(\bs g_k,\psi_{j})]]+a^2\mathbb{E}_k[\abs{(\psi_{j},\bs g_k)}^2].
\end{equation}Let's estimate the middle and last terms. First, by definition \eqref{eq: semi-unbiased}, the middle term reduces to
\begin{equation}
    (\psi_{j},\bs x_k)\mathbb{E}_k[(\bs g_k,\psi_{j})]=\frac{\lambda_j}{N}\abs{(\psi_{j},\bs x_k)}^2.
\end{equation}To estimate the last term, we take two steps. First we simplify the following outer product
\begin{equation}
    \begin{split}
        \mathbb{E}_k[\bs g_k\bs g_k^*] &= \mathbb{E}_k\left[(H\bs x_k,\bs e_{i_k})(\bs e_{i_k},H\bs x_k)\bs e_{i_k}\bs e_{i_k}^*\right] = \frac{1}{N}\mathrm{diag}(H\bs x_k\bs x_k^*H)
    \end{split} 
\end{equation}Then, the last term can be simplifed to
\begin{equation}
    \mathbb{E}_k[\abs{(\psi_{j},\bs g_k)}^2]=\frac{1}{N}\psi_{j}^*\mathrm{diag}\left(H\bs x_k\bs x_k^*H\right)\psi_{j}.
\end{equation}To put these all together, we have
\begin{equation}
    \begin{split}\label{eq: eachinnerpd}
         \mathbb{E}_k[\abs{(\psi_{j},\bs x_{k+1})}^2]=b_j\abs{(\psi_{j},\bs x_k)}^2 + \frac{a^2}{N}\psi_{j}^*\mathrm{diag}\left(H\bs x_k\bs x_k^*H\right)\psi_{j},
    \end{split}
\end{equation}where 
\begin{equation}\label{eq: tau}
    b_j=1-\frac{2}{N}\lambda_ja.
\end{equation}Notice that 
\begin{equation}\label{eq: monotone b}
    b_j\geq b_{j+1}> 0
\end{equation}whenever $a>0$ with $a$ being sufficiently small.
From \eqref{eq: eachinnerpd}, we observe that
\begin{equation}\label{eq: ubnorm}
\begin{split}
    &\mathbb{E}_k[\norm{\bs x_{k+1}}^2]\\
    &=\sum_{j=1}^N\mathbb{E}_k[\abs{(\psi_{j},\bs x_{k+1})}^2]\\
    &=\sum_{j=1}^Nb_j\abs{(\psi_{j},\bs x_k)}^2+\frac{a^2}{N}\sum_{j=1}^N\mathrm{tr}\left(\psi_{j}^*\mathrm{diag}\left(H\bs x_k\bs x_k^*H\right)\psi_{j}\right)\\
    &=\sum_{j}b_j\abs{(\psi_{j},\bs x_k)}^2+\frac{a^2}{N}\mathrm{tr}\left(\mathrm{diag}\left(H\bs x_k\bs x_k^*H\right)\sum_j\psi_{j}\psi_{j}^*\right)\\
    &=\sum_{j}b_j\abs{(\psi_{j},\bs x_k)}^2+\frac{a^2}{N}\mathrm{tr}\left(H\bs x_k\bs x_k^*H\right)\\
    &= \sum_{j}b_j\abs{(\psi_{j},\bs x_k)}^2 + \frac{a^2}{N}\norm{H\bs x_k}^2.
\end{split}
\end{equation}Observe that
\begin{equation}\label{eq: residual}
    \norm{H\bs x_k}^2 = \sum_{j=1}^N\lambda_j^2\abs{(\psi_{j},\bs x_k)}^2\leq \lambda_1^2\abs{(\psi_{GS},\bs x_k)}^2+\max_{j\neq 1}\lambda_j^2\left(\norm{\bs x_k}^2-\abs{(\psi_{GS},\bs x_k)}^2\right).
\end{equation}
With these observations,  we derive a recursive relation from the inequality \eqref{eq: recursive ineq} as,
\begin{equation}
    \begin{split}\label{eq: recursive ineq 2}
        \mathbb{E}_k\left[\frac{\abs{(\psi_{GS},\bs x_{k+1})}^2}{\|\bs x_{k+1}\|^2}\right]&\geq \frac{\abs{\mathbb{E}_k[(\psi_{GS},\bs x_{k+1})]}^2}{\mathbb{E}_k[\norm{\bs x_{k+1}}^2]}\\
        &= \frac{\left(1-\frac{\lambda_1}{N}a\right)^2\abs{(\psi_{GS},\bs x_k)}^2}{\sum_{j}b_j\abs{(\psi_{j},\bs x_k)}^2 + \frac{a^2}{N}\norm{H\bs x_k}^2}\quad\eqref{eq: nominator}, \eqref{eq: ubnorm}\\
        &\geq \frac{\left(1-\frac{\lambda_1}{N}a\right)^2\abs{(\psi_{GS},\bs x_k)}^2}{b_1\abs{(\psi_{GS},\bs x_k)}^2 + b_2\sum_{j\neq 1}\abs{(\psi_{j},\bs x_k)}^2 + \frac{a^2}{N}\norm{H\bs x_k}^2}\quad \eqref{eq: monotone b} \\
        &= \frac{\left(1-\frac{\lambda_1}{N}a\right)^2\abs{(\psi_{GS},\bs x_k)}^2}{b_1\abs{(\psi_{GS},\bs x_k)}^2 + b_2(\norm{\bm x_k}^2-\abs{(\psi_{GS},\bs x_k)}^2) + \frac{a^2}{N}\norm{H\bs x_k}^2}\\
        &\geq \frac{\left(1-\frac{\lambda_1}{N}a\right)^2\abs{(\psi_{GS},\bs x_k)}^2}{(b_1+\frac{a^2}{N}\lambda_1^2)\abs{(\psi_{GS},\bs x_k)}^2 + (b_2+\frac{a^2}{N}\max_{j\neq 1}\lambda_j^2)(\norm{\bm x_k}^2-\abs{(\psi_{GS},\bs x_k)}^2)}\quad \eqref{eq: residual}\\
        &\geq \frac{A\frac{\abs{(\psi_{GS},\bs x_k)}^2}{\norm{x_k}^2}}{A\frac{\abs{(\psi_{GS},\bs x_k)}^2}{\norm{x_k}^2}+B\left(1-\frac{\abs{(\psi_{GS},\bs x_k)}^2}{\norm{x_k}^2}\right)+\frac{a^2(N-1)}{N^2}\lambda_1^2}\quad \eqref{eq: tau},\\
        &\frac{A\frac{\abs{(\psi_{GS},\bs x_k)}^2}{\norm{x_k}^2}}{A\frac{\abs{(\psi_{GS},\bs x_k)}^2}{\norm{x_k}^2}+B\left(1-\frac{\abs{(\psi_{GS},\bs x_k)}^2}{\norm{x_k}^2}\right)+\frac{a^2}{N}\lambda_1^2},
    \end{split}
\end{equation}where 
\begin{subequations}\label{eq: coeffs}
    \begin{eqnarray}  &A = \left(1-\frac{\lambda_1}{N}a\right)^2\\
    &B = b_2+\frac{a^2}{N}\max_{j\neq 1}\lambda_j^2.
    \end{eqnarray}
\end{subequations}Let us denote $r_k:=\frac{\abs{(\psi_{GS},\bs x_{k})}^2}{\|\bs x_{k}\|^2}$. Note that $r_k\in[0,1]$ for any $k$ by definition. With this notation, the above inequality is rewritten as
\begin{equation}\label{eq: original reeq}
    \mathbb{E}_k[r_{k+1}]\geq \frac{A}{Ar_k+B(1-r_k)+\frac{a^2}{N}\lambda_1^2}r_k.
\end{equation}Under the spectral gap assumption that $\lambda_1<\lambda_2$, we observe that
\begin{equation}\label{eq: A-B}
\begin{split}
    A-B&=\frac{2a}{N}(\lambda_2-\lambda_1)-\frac{a^2}{N}(\max_{j\neq 1}\lambda_j^2-\frac{\lambda_1^2}{N})>\frac{a}{N}(\lambda_2-\lambda_1)>0
\end{split}
\end{equation}for any $a\in\left(0,\frac{\lambda_2-\lambda_1}{\abs{\max_{j\neq 1}\lambda_j^2-\frac{\lambda_1^2}{N}}}\right)$. 

For $\epsilon>0$, we consider the event that $r_k\leq 1-\epsilon$. In this event, if we set $a\leq\frac{(\lambda_2-\lambda_1)\epsilon}{2\lambda_1^2}$, then it follows from \eqref{eq: A-B},
\begin{equation}
    a^2<\frac{(\lambda_2-\lambda_1)a\epsilon}{2\lambda_1^2}<\frac{N(A-B)\epsilon}{2\lambda_1^2}.
\end{equation}From this, the ratio in \eqref{eq: original reeq} satisfies
\begin{equation}\label{eq: A,B}
\begin{split}
    \frac{A}{Ar_k+B(1-r_k)+\frac{a^2}{N}\lambda_1^2}&\geq \frac{A}{A(1-\epsilon)+B\epsilon+\frac{a^2}{N}\lambda_1^2}\\
    &>\frac{A}{A(1-\epsilon)+B\epsilon+\frac{(A-B)\epsilon}{2}}\\
    &=\frac{A}{A(1-\frac{\epsilon}{2})+B\frac{\epsilon}{2}}>1,
\end{split}
\end{equation}since $B<A$. Denote the ratio as
\begin{equation}\label{cepsil}
    c_{\epsilon}:=\frac{A}{A(1-\frac{\epsilon}{2})+B\frac{\epsilon}{2}}.
\end{equation}Let $R_k:=r_{k}\mathbb{I}_{k}$ where $\mathbb{I}_k$ values one if $r_j\leq 1-\epsilon$ for all $j\leq k$ and zero otherwise. By this definition, it follows that conditioned on $\mathcal{F}_k$, where the $\sigma$-algebra $\mathcal{F}_k$ is associated with the stochastic process up to $k$, 
\begin{equation}
    \mathbb{E}_k[R_{k+1}]\geq c_\epsilon R_k.
\end{equation}This implies 
\begin{equation}
    \mathbb{E}[R_T]\geq c_\epsilon^{T-1}\abs{\bra{\psi_1}\ket{\bm x_1}}^2.
\end{equation}For $T=\mathcal{O}\left(\log_{c_\epsilon}\frac{1-\epsilon}{\abs{\bra{\psi_1}\ket{\bm x_1}}^2}\right)$, it should be that $R_T=1-\epsilon$ with probability one, since $R_k\leq 1-\epsilon$ for all $k$ by definition. Therefore, this implies that $r_T\geq 1-\epsilon$ with probability one conditioned on that $r_k<1-\epsilon$ for all $k<T$, or it could happen that $r_t\geq 1-\epsilon$ for some $t<T$. As a result, the worst-case iteration number scales as
\begin{equation}\label{eq: worst-ineq}
    \frac{1}{\log c_\epsilon}=\frac{1}{\log \left(1+\frac{(A-B)\epsilon}{A-\epsilon(A-B)}\right)}\leq c\frac{A-\epsilon(A-B)}{(A-B)\epsilon}\leq \frac{c}{(A-B)\epsilon} =\mathcal{O}\left( \frac{N}{(\lambda_2-\lambda_1)a\epsilon}\right)=\mathcal{O}\left( \frac{N\lambda_1^2}{(\lambda_2-\lambda_1)\epsilon^{2}}\right),
\end{equation}for some absolute constant $c$, since $\frac{1}{A-B}=\mathcal{O}\left(\frac{N}{a(\lambda_2-\lambda_1)}\right)$ by \eqref{eq: A-B} and $a\leq \frac{(\lambda_2-\lambda_1)\epsilon}{2\lambda_1^2}$.

For the case of degenerate ground states, let $\{\psi_{GS,j}\}_{j=1}^g$ be the ground states and $\lambda_1=\cdots=\lambda_g<\lambda_{g+1}\leq\cdots\lambda_N$. By summing \eqref{eq: recursive ineq 2} over the ground states, we achieve that
\begin{equation}
    \begin{split}
        \mathbb{E}_k\left[\frac{\sum_{j=1}^g\abs{(\psi_{GS,j},\bs x_{k+1})}^2}{\|\bs x_{k+1}\|^2}\right]&\geq \frac{\sum_{j=1}^g\abs{\mathbb{E}_k[(\psi_{GS,j},\bs x_{k+1})]}^2}{\mathbb{E}_k[\norm{\bs x_{k+1}}^2]}\\
        &=\frac{\left(1-\frac{\lambda_1}{N}a\right)^2\sum_{j=1}^g\abs{(\psi_{GS,j},\bs x_k)}^2}{\sum_{j=1}^Nb_j\abs{(\psi_{j},\bs x_k)}^2 + \frac{a^2}{N}\norm{H\bs x_k}^2}\\
        &\geq \frac{A\frac{\sum_{j=1}^g\abs{(\psi_{GS,j},\bs x_k)}^2}{\norm{\bs x_k}^2}}{A\frac{\sum_{j=1}^g\abs{(\psi_{GS,j},\bs x_k)}^2}{\norm{\bs x_k}^2}+B\left(1-\frac{\sum_{j=1}^g\abs{(\psi_{GS,j},\bs x_k)}^2}{\norm{\bs x_k}^2}\right)+\frac{a^2(N-1)}{N^2}\lambda_1^2},
    \end{split}
\end{equation}where \begin{subequations}
    \begin{eqnarray}  &A = \left(1-\frac{\lambda_1}{N}a\right)^2\\
    &B = b_{g+1}+\frac{a^2}{N}\max_{j\geq g+1}\lambda_j^2,
    \end{eqnarray}
\end{subequations}and $b_{j}$ is defined in \eqref{eq: tau}. By denoting $r_k=\frac{\sum_{j=1}^g\abs{(\psi_{GS,j},\bs x_k)}^2}{\norm{\bs x_k}^2}$, we can obtain a similar recursive inequality as \eqref{eq: original reeq}. Additionally, from the assumption that $\lambda_1=\cdots=\lambda_g<\lambda_{g+1}$, we obtain
\begin{equation}
    A-B = \frac{2a}{N}(\lambda_{g+1}-\lambda_1)-\frac{a^2}{N}(\max_{j\geq g+1}\lambda_j^2-\frac{\lambda_1^2}{N})>\frac{a}{N}(\lambda_{g+1}-\lambda_1)>0
\end{equation}for any $a\in\left(0,\frac{\lambda_{g+1}-\lambda_1}{\abs{\max_{j\geq g+1}\lambda_j^2-\frac{\lambda_1^2}{N}}}\right)$. Therefore, the proof is complete similar to the above proof.

\end{proof}

\subsection{The proof of \cref{thm: convergence2}}\label[appsec]{sec: proof of convergence thm2}We prove \cref{thm: convergence2} which generalizes \cref{thm: convergence} to the case $m_r<N$.  We consider the estimator \eqref{unbiasedest} and show convergence of \cref{alg: algorithm1}.
By \cref{proposition}, we notice that
\begin{equation}\label{eq: nominator2}
\begin{split}
    &\mathbb{E}[(\psi_{GS},\bs x_{k+1})]\\
    &=(\psi_{GS},\mathbb{E}[\bs x_{k+1}])\\
    &=(\psi_{GS}, \bs x_k) - a(\psi_{GS},\mathbb{E}[\bs g_k])\\
    &=(\psi_{GS},\bs x_k) - a(\psi_{GS},\frac{m_rm_c}{N^2}H\bs x_k)\\
    &=\left(1-\frac{m_rm_c\lambda_1}{N^2}a\right)(\psi_{GS},\bs x_k),
\end{split} 
\end{equation}and
\begin{equation}
    \mathbb{E}_k[(\psi_{j},\bs x_k)(\bs g_k,\psi_{j})]=(\psi_{j}, \bs x_k)\mathbb{E}_k[(\bs g_k,\psi_{j})]=(\psi_{j}, \bs x_k)(\frac{m_rm_c}{N^2}H\bs x_k,\psi_{j})=\frac{m_rm_c\lambda_j}{N^2}\abs{(\psi_{j},\bs x_k)}^2.
\end{equation}Additionally, it follows that
\begin{equation}
\begin{split}
    &\mathbb{E}_k[\abs{(\psi_{j},\bs g_k)}^2] = \frac{m_c(N-m_c)}{N(N-1)}\psi_{j}^*\mathrm{diag}\left(H\bs x_k\bs x_k^*H\right)\psi_{j} +\frac{\lambda_j^2m_c(m_c-1)}{N(N-1)}\abs{(\psi_{j},\bs x_k)}^2+\psi_{j}^*\Sigma_{r,c}\psi_{j},
\end{split}
\end{equation}and as we did in \eqref{eq: ubnorm}, we achieve that
\begin{equation}
    \mathbb{E}_k[\norm{\bs x_{k+1}}^2]=\sum_{j}b_j\abs{(\psi_{j},\bs x_k)}^2 +\frac{m_c(N-m_c)}{N(N-1)}a^2\norm{H\bs x_k}^2+a^2\mathrm{tr}(\Sigma_{r,c}),    
\end{equation}where $b_j=1-\frac{2a\lambda_jm_rm_c}{N^2}+a^2\frac{\lambda_j^2m_c(m_c-1)}{N(N-1)}$ and $\Sigma_{r,c}$ is defined in \cref{proposition}. Then, we have
\begin{equation}\label{eq: recursive ineq 3}
    \begin{split}
        \mathbb{E}_k\left[\frac{\abs{(\psi_{GS},\bs x_{k+1})}^2}{\|\bs x_{k+1}\|^2}\right]&\geq \frac{\abs{\mathbb{E}_k[(\psi_{GS},\bs x_{k+1})]}^2}{\mathbb{E}_k[\norm{\bs x_{k+1}}^2]}\\
        &= \frac{\left(1-\frac{m_rm_c\lambda_1}{N^2}a\right)^2\abs{(\psi_{GS},\bs x_k)}^2}{\sum_{j}b_j\abs{(\psi_{j},\bs x_k)}^2 +\frac{m_c(N-m_c)}{N(N-1)}a^2\norm{H\bs x_k}^2+a^2\mathrm{tr}(\Sigma_{r,c})}\\
        &\geq \frac{\left(1-\frac{m_rm_c\lambda_1}{N^2}a\right)^2\abs{(\psi_{GS},\bs x_k)}^2}{\sum_{j}b_j\abs{(\psi_{j},\bs x_k)}^2 +\left[\frac{m_c(N-m_c)}{N(N-1)}\norm{H}^2\norm{\bs x_k}^2+\mathrm{tr}(\Sigma_{r,c})\right]a^2}\\
        &\geq \frac{\left(1-\frac{m_rm_c\lambda_1}{N^2}a\right)^2\abs{(\psi_{GS},\bs x_k)}^2}{b_1\abs{(\psi_{GS},\bs x_k)}^2 + b_2(\norm{\bm x_k}^2-\abs{(\psi_{GS},\bs x_k)}^2) +\left[\frac{m_c(N-m_c)}{N(N-1)}\norm{H}^2\norm{\bs x_k}^2+\mathrm{tr}(\Sigma_{r,c})\right]a^2}\\
        &\geq \frac{A\frac{\abs{(\psi_{GS},\bs x_k)}^2}{\norm{\bs x_k}^2}}{A\frac{\abs{(\psi_{GS},\bs x_k)}^2}{\norm{\bs x_k}^2}+B\left(1-\frac{\abs{(\psi_{GS},\bs x_k)}^2}{\norm{\bs x_k}^2}\right)+\left[\frac{-A+b_1}{a^2}+\frac{m_c(N-m_c)}{N(N-1)}\norm{H}^2+\frac{\mathrm{tr}(\Sigma_{r,c})}{\norm{\bm x_k}^2}\right]a^2}\\
        &\geq \frac{A\frac{\abs{(\psi_{GS},\bs x_k)}^2}{\norm{\bs x_k}^2}}{A\frac{\abs{(\psi_{GS},\bs x_k)}^2}{\norm{\bs x_k}^2}+B\left(1-\frac{\abs{(\psi_{GS},\bs x_k)}^2}{\norm{\bs x_k}^2}\right)+\frac{3m_c}{N}\norm{H}^2a^2}
    \end{split}
\end{equation}where \begin{subequations}\label{eq: coeffs3}
    \begin{eqnarray}  &A = \left(1-\frac{m_rm_c\lambda_1}{N^2}a\right)^2\\
    &B = b_2.
    \end{eqnarray}
\end{subequations}The third inequality holds since $\norm{H\bm x_k}^2\leq \norm{H}^2\norm{\bm x_k}^2$. The second last inequality holds since $\abs{(\psi_{GS},\bs x_k)}^2\leq\norm{\bm x_k}^2$. 
The last inequality follows from the assumption that $N\gg \lambda_1,m_r,m_c,\norm{H}$, which yields
\begin{equation}
\begin{split}
     &\frac{-A+b_1}{a^2}+\frac{m_c(N-m_c)}{N(N-1)}\norm{H}^2+\frac{\mathrm{tr}(\Sigma_{r,c})}{\norm{\bm x_k}^2}\\
     &\leq-\frac{m_r^2m_c^2\lambda_1^2}{N^4}+\frac{m_c(m_c-1)}{N(N-1)}\lambda_1^2+\frac{m_c(N-m_c)}{N(N-1)}\norm{H}^2+\frac{m_c(N-m_r)}{N^2}\norm{H}^2\\
     &\leq \frac{3m_c}{N}\norm{H}^2.
\end{split} 
\end{equation}by \cref{proposition}.
In this case, the difference between these two constants is
\begin{equation}
\begin{split}
    A-B &= \frac{2m_rm_ca}{N^2}(\lambda_2-\lambda_1)+\left(\frac{m_r^2m_c^2\lambda_1^2}{N^4}-\frac{\lambda_2^2m_c(m_c-1)}{N(N-1)}\right)a^2\\
    &\geq \frac{m_rm_ca}{N^2}(\lambda_2-\lambda_1)+\frac{m_r^2m_c^2\lambda_1^2}{N^4}a^2>0,
\end{split}
\end{equation}for any $a\in\left(0,\frac{2m_r(\lambda_2-\lambda_1)}{\lambda_2^2m_c}\right)$. Following after \eqref{eq: original reeq} in \cref{sec: proof of convergence thm}, we complete the proof of \cref{thm: convergence2} similarly  by setting $a<\frac{m_r\epsilon}{6N\norm{H}^2}$.  Lastly, we notice that for any $a<\frac{1}{\norm{H}}$, $\norm{\bm x_k}>0$ for all $k$ as we derived.

Similar to \eqref{eq: worst-ineq}, the worst-case iteration number scales as
\begin{equation}
    \frac{1}{(A-B)\epsilon}=\mathcal{O}\left( \frac{N^2}{m_rm_ca(\lambda_2-\lambda_1)\epsilon}\right)=\mathcal{O}\left(\frac{N^3}{m_r^2m_c\epsilon^2}\right).
\end{equation}
This cubic-scaling complexity with respect to $N$ can be reduced to a linear scaling if we consider the sparsity of $H$ when sampling the indices in \eqref{eq: esimatemat}.  First, we observe that the first property in \cref{proposition} is modified to 
\begin{equation}
     \mathbb{E}[\bs g] =\frac{m_rm_c}{Ns}H\bs x.
\end{equation}Similar to the proof of \cref{proposition}, another modification is made as 
\begin{equation}\label{eq: modified variance}
    \mathbb{E}[\bs g\bs g^*]=\frac{m_c(N-m_c)}{N(N-1)}\mathrm{diag}\left(H\bs x \bs x^*H\right)+\frac{m_c(m_c-1)}{N(N-1)}H\bs x\bs x^*H+\Sigma_{r,c},
\end{equation}where 
\begin{equation}
    \begin{split}
        \mathrm{tr}(\Sigma_{r,c})&= \mathbb{E}_c\left[\mathbb{E}_r\left[\sum_{i_c}\abs{(\bs x_{r,c}-\bs x,H\bs e_{i_c})}^2\right]\right]\\
    &=\mathbb{E}_c\left[\sum_{i_c}\mathbb{E}_r\left[\abs{(\bs x_{r,c}-\bs x,H\bs e_{i_c})}^2\right]\right]\\
    &=\frac{m_c}{N}\sum_{i=1}^N\mathbb{E}_r\left[\abs{(\bs x_{r,i}-\bs x,H\bs e_{i})}^2\right]\\
    &\leq \frac{m_c(s-m_r)}{Ns}\norm{H}^2\norm{\bm x}^2.
    \end{split}
\end{equation}
Compared to \cref{proposition}, we notice the change of the last term in \eqref{eq: modified variance} such that its upper bound has the reduced factor, $\frac{s-m_r}{s}$, instead $\frac{N-m_r}{N}$. This can be similarly checked as \eqref{eq: Sigmarc prop} by considering the sparsity of $H$. With this result in mind,  we can modify the inequality \eqref{eq: recursive ineq 3} as
\begin{equation}\label{eq: recursive ineq 4}
    \begin{split}
        \mathbb{E}_k\left[\frac{\abs{(\psi_{GS},\bs x_{k+1})}^2}{\|\bs x_{k+1}\|^2}\right]\geq \frac{A\frac{\abs{(\psi_{GS},\bs x_k)}^2}{\norm{\bs x_k}^2}}{A\frac{\abs{(\psi_{GS},\bs x_k)}^2}{\norm{\bs x_k}^2}+B\left(1-\frac{\abs{(\psi_{GS},\bs x_k)}^2}{\norm{\bs x_k}^2}\right)+\frac{3m_c}{N}\norm{H}^2a^2},
    \end{split}
\end{equation}where
\begin{subequations}\label{eq: coeffs4}
    \begin{eqnarray}  &A = \left(1-\frac{m_rm_c\lambda_1}{Ns}a\right)^2\\
    &B = b_2,
    \end{eqnarray}
\end{subequations}and for each $j\in[N]$,
\begin{equation}
    b_j=1-\frac{2a\lambda_jm_rm_c}{Ns}+a^2\frac{\lambda_j^2m_c(m_c-1)}{N(N-1)}.
\end{equation}As a result, we have
\begin{equation}
    A-B = \frac{2a(\lambda_2-\lambda_1)m_rm_c}{Ns}+\frac{m_r^2m_c^2\lambda_1^2a^2}{N^2s^2}-\frac{\lambda_2^2m_c(m_c-1)a^2}{N(N-1)}\geq \frac{a(\lambda_2-\lambda_1)m_rm_c}{Ns},
\end{equation}for any $a\in(0, \frac{Nm_r(\lambda_2-\lambda_1)}{s\lambda_2^2}]$. In addition, for any $a\in(0,\frac{(\lambda_2-\lambda_1)m_r\epsilon}{6s\norm{H}^2}]$, it follows that $\frac{3m_c\norm{H}^2a^2}{N}\leq \frac{(A-B)\epsilon}{2}$. As a result, the worst-case iteration number scales as
\begin{equation}
    \frac{1}{(A-B)\epsilon}=\mathcal{O}\left( \frac{Ns}{m_rm_ca(\lambda_2-\lambda_1)\epsilon}\right)=\mathcal{O}\left(\frac{Ns^2}{m_r^2m_c\epsilon^2}\right),
\end{equation}since $a\leq \frac{(\lambda_2-\lambda_1)m_r\epsilon}{6s\norm{H}^2}$.

\subsection{Additional numerical result}\label[appsec]{sec: additional results}

\cref{fig:ex2} plots 9 simulation results among the 100 ones in \cref{fig:ex1}, where the two fidelities of the iteration with two degenerate ground states are shown.
\begin{figure}[htbp]
    \centering
    \begin{subfigure}[b]{0.3\textwidth}
    \includegraphics[width=\textwidth]{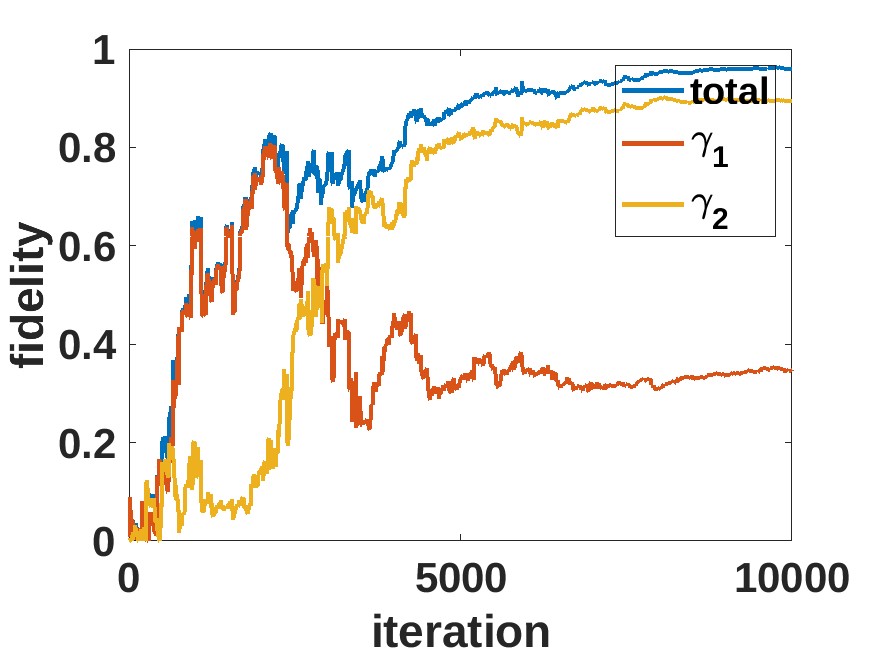}
    \end{subfigure}
    \begin{subfigure}[b]{0.3\textwidth}
    \includegraphics[width=\textwidth]{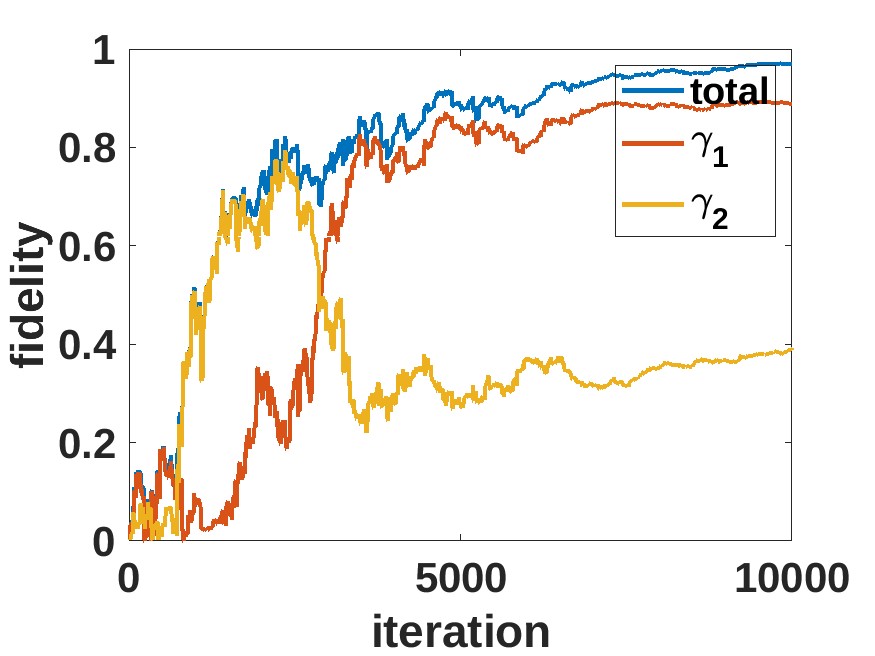}
    \end{subfigure}
    \begin{subfigure}[b]{0.3\textwidth}
    \includegraphics[width=\textwidth]{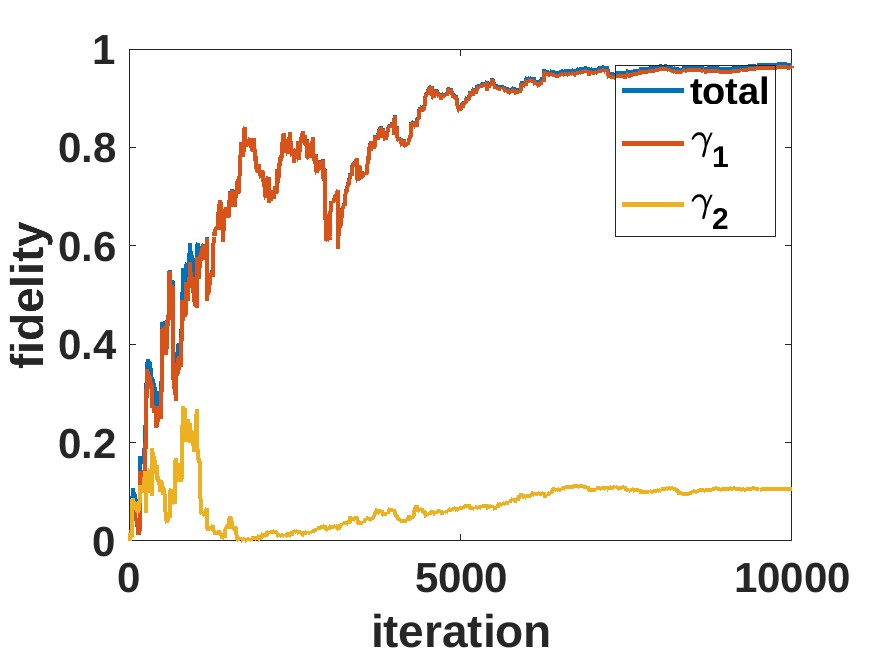}
    \end{subfigure}

    \begin{subfigure}[b]{0.3\textwidth}
    \includegraphics[width=\textwidth]{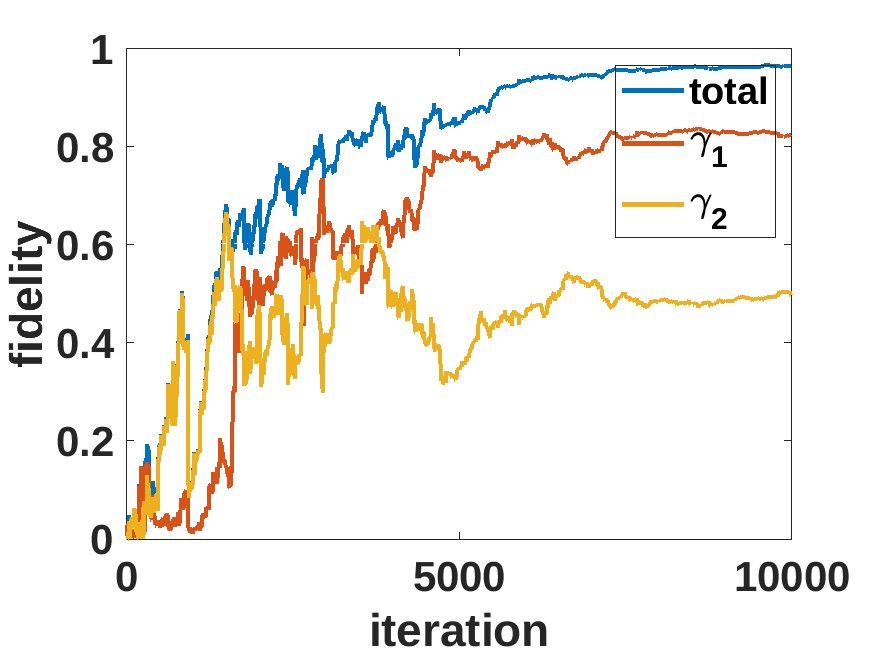}
    \end{subfigure}
    \begin{subfigure}[b]{0.3\textwidth}
    \includegraphics[width=\textwidth]{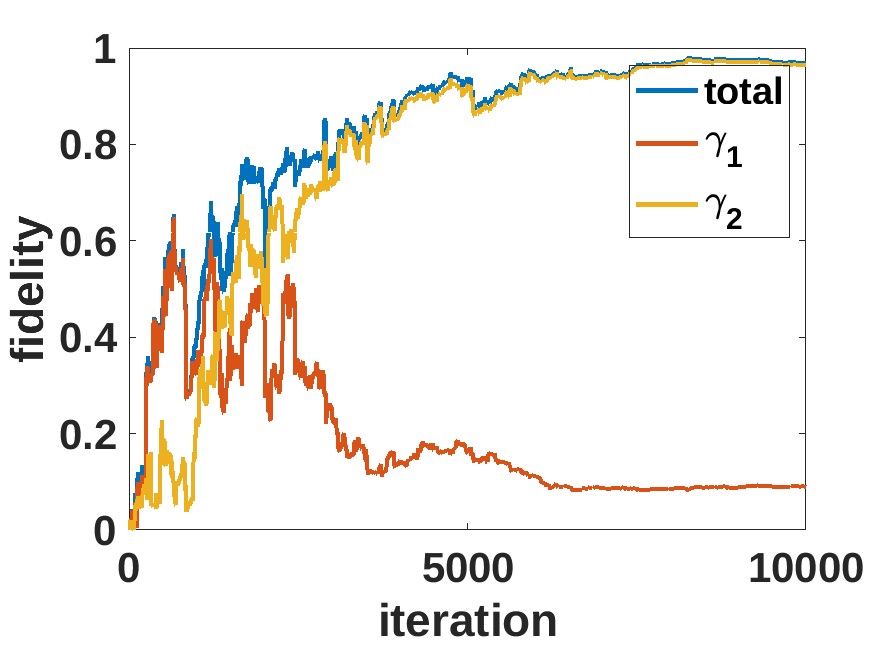}
    \end{subfigure}
    \begin{subfigure}[b]{0.3\textwidth}
    \includegraphics[width=\textwidth]{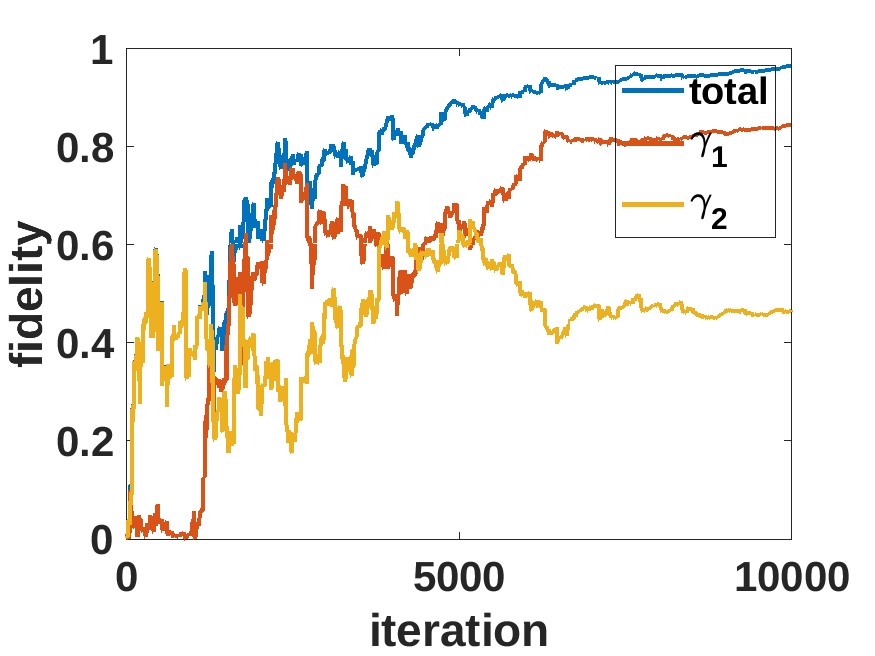}
    \end{subfigure}

    \begin{subfigure}[b]{0.3\textwidth}
    \includegraphics[width=\textwidth]{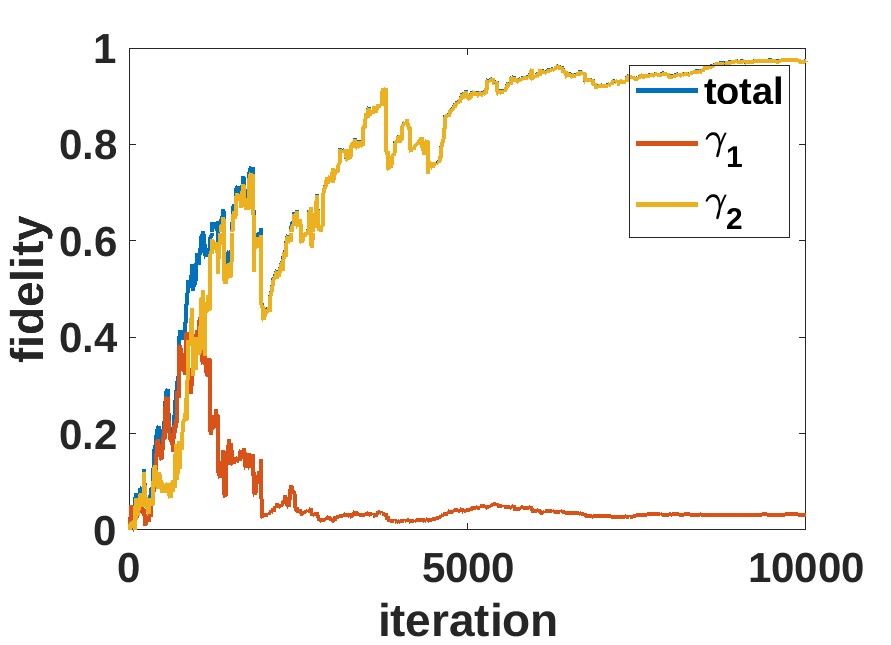}
    \end{subfigure}
    \begin{subfigure}[b]{0.3\textwidth}
    \includegraphics[width=\textwidth]{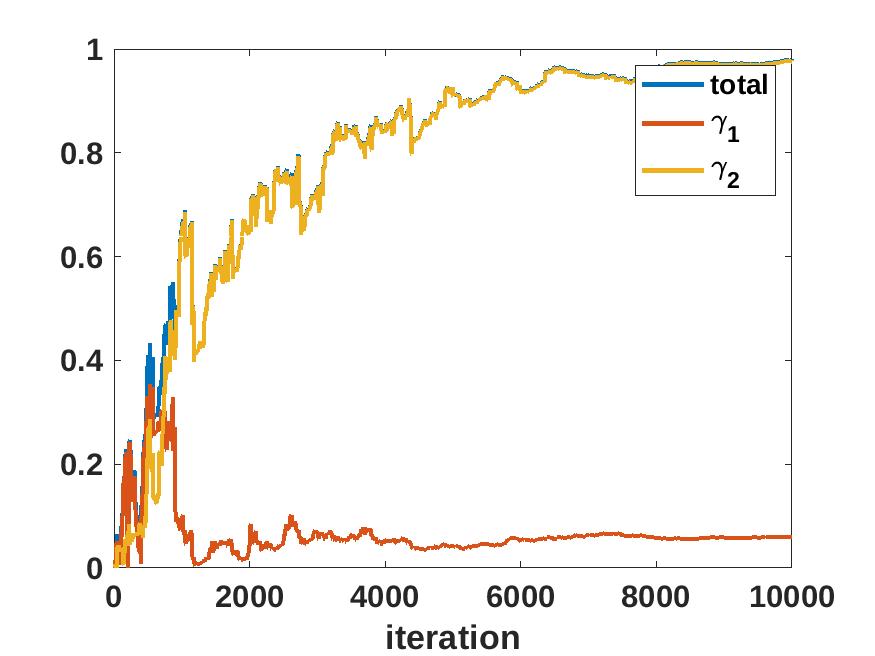}
    \end{subfigure}
    \begin{subfigure}[b]{0.3\textwidth}
    \includegraphics[width=\textwidth]{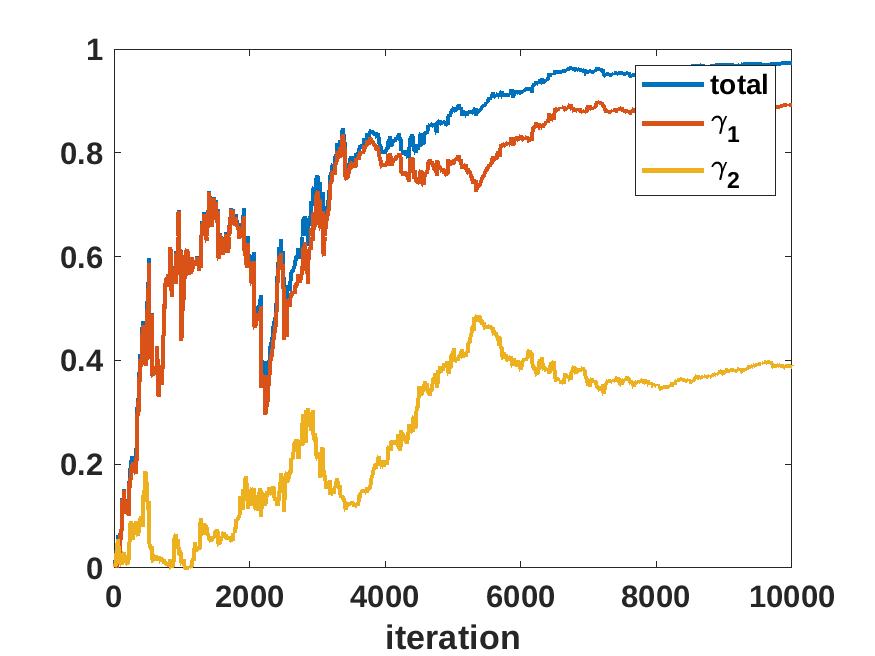}
    \end{subfigure}

    \begin{subfigure}[b]{0.3\textwidth}
    \includegraphics[width=\textwidth]{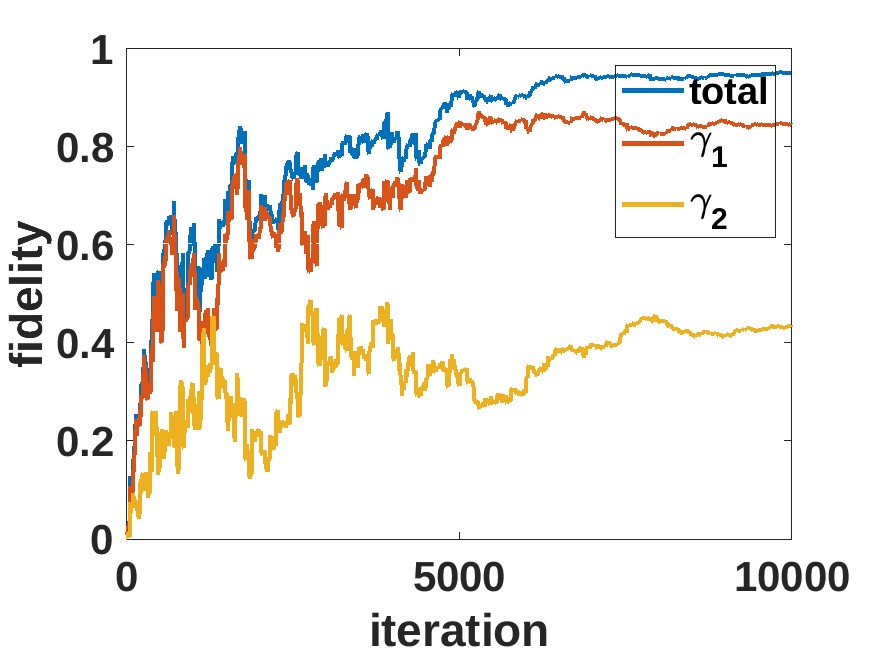}
    \end{subfigure}
    \begin{subfigure}[b]{0.3\textwidth}
    \includegraphics[width=\textwidth]{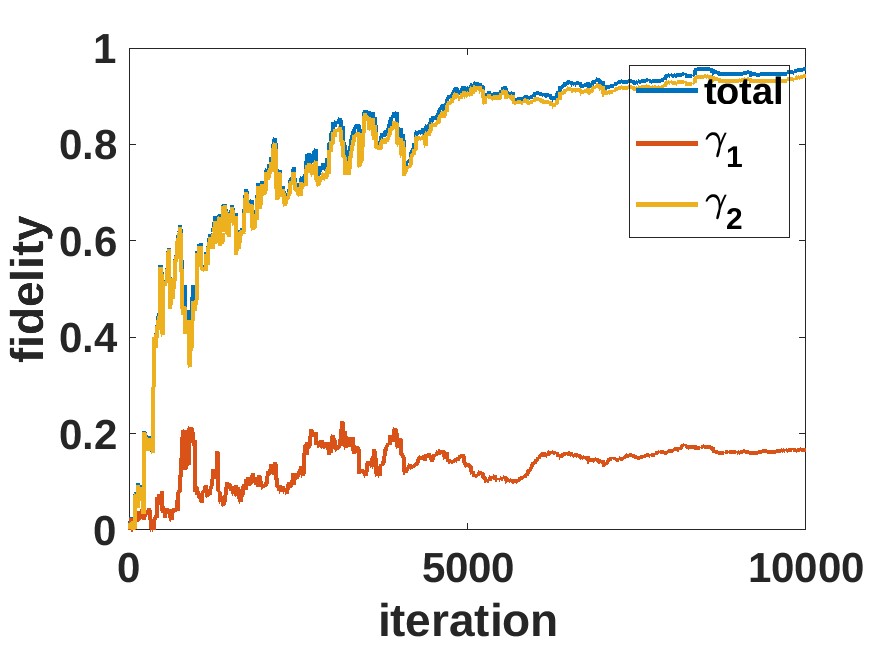}
    \end{subfigure}
    \begin{subfigure}[b]{0.3\textwidth}
    \includegraphics[width=\textwidth]{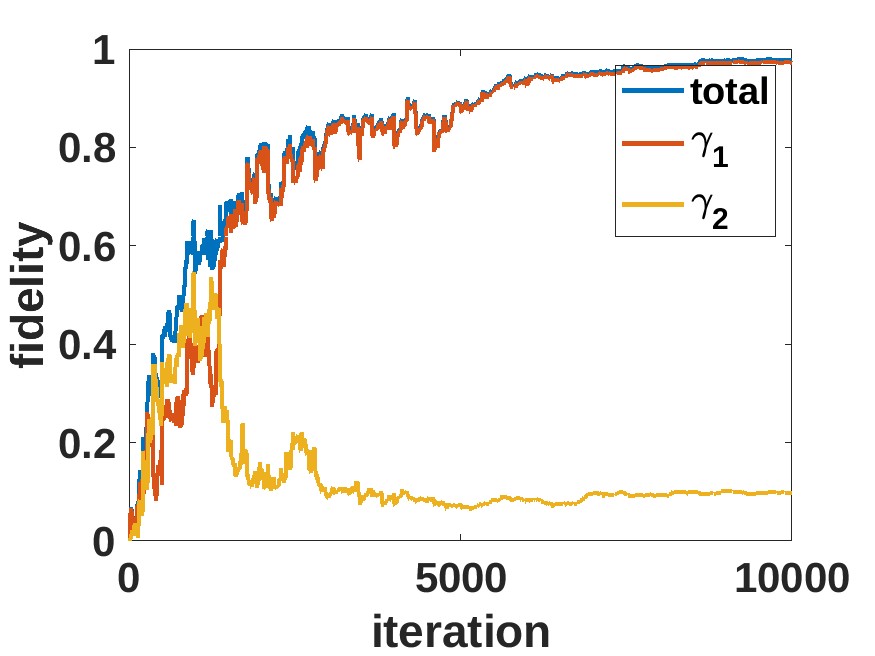}
    \end{subfigure}
    \caption{Convergence of fidelities, $\gamma_1=\abs{\frac{(\psi_{GS,1},\bs x)}{\norm{\bm x}}}$, $\gamma_2=\abs{\frac{(\psi_{GS,2},\bs x)}{\norm{\bm x}}}$, and $\sqrt{\gamma_1^2+\gamma_2^2}$, which is observed from 12 independent simulations among the 100 ones in \cref{fig:ex1}. 
    Results named "total" correspond to the fidelity \eqref{eq: fidelity} of the iteration with the two ground states.}
    \label{fig:ex2}
\end{figure}

\bibliographystyle{plain}
\bibliography{qc}

\end{document}